\documentclass[a4paper, authoryear, 12pt]{elsarticle}

\usepackage[english]{babel}
\usepackage{babel}
\usepackage{fontenc}
\usepackage[table]{xcolor}
\usepackage{boldline} 
\usepackage{graphicx}
\usepackage{caption}
\usepackage{booktabs}
\usepackage{amsmath,amsthm,amssymb}
\usepackage{natbib}
\usepackage[colorlinks=true,citecolor=blue]{hyperref}
\usepackage{color}
\usepackage{graphicx}
\usepackage{subfigure}
\usepackage{color}
\usepackage{newlfont}
\usepackage{multirow}
\usepackage{longtable} 
\usepackage{ifthen}
\usepackage{alltt}
\usepackage{enumerate}
\usepackage[text={6.25in,8.5in},centering]{geometry}
\usepackage{tikz}
\usepackage{epstopdf}
\usepackage{float}
\usepackage{arydshln}

\usepackage{tikz}
\usetikzlibrary{shapes,arrows}
\numberwithin{equation}{section}
\newtheorem{theorem}{Theorem}[section]
\newtheorem{lemma}{Lemma}[section]
\newtheorem{proposition}{\bf Proposition}[section]

\newtheorem{corollary}{\hskip\parindent\bf Corollary}[section]

\input epsf.sty
\usepackage{lineno}
\journal{Elsevier}

\begin{document}
	\title{Assessment of Lockdown Effect in Some States and Overall India: A Predictive Mathematical Study on COVID-19 Outbreak}

	\author[DAC]{Tridip Sardar \footnote{Corresponding author. Email: tridipiitk@gmail.com}}
	\author[ISI]{Sk Shahid Nadim }
	\author[VB]{Sourav Rana}
	\author[ISI]{Joydev Chattopadhyay}
	\address[DAC]{Department of Mathematics, Dinabandhu Andrews College, Kolkata, India}
	\address[ISI]{Agricultural and Ecological Research Unit, Indian Statistical Institute, Kolkata, India}
	\address[VB]{Department of Statistics, Visva-Bharati University, Santiniketan, West Bengal, India}
	
\begin{abstract}
In the absence of neither an effective treatment or vaccine and with an incomplete understanding of the epidemiological cycle, Govt. has implemented a nationwide lockdown to reduce COVID-19 transmission in India. To study the effect of social distancing measure, we considered a new mathematical model on COVID-19 that incorporates lockdown effect. By validating our model to the data on notified cases from five different states and overall India, we estimated several epidemiologically important parameters as well as the basic reproduction number ($R_{0}$). Combining the mechanistic mathematical model with different statistical forecast models, we projected notified cases in the six locations for the period May 17, 2020, till May 31, 2020. A global sensitivity analysis is carried out to determine the correlation of two epidemiologically measurable parameters on the lockdown effect and also on $R_{0}$. Our result suggests that lockdown will be effective in those locations where a higher percentage of symptomatic infection exists in the population. Furthermore, a large scale COVID-19 mass testing is required to reduce community infection. Ensemble model forecast suggested a high rise in the COVID-19 notified cases in most of the locations in the coming days. Furthermore, the trend of the effective reproduction number ($R_{t}$) during the projection period indicates if the lockdown measures are completely removed after May 17, 2020, a high spike in notified cases may be seen in those locations. Finally, combining our results, we provided an effective lockdown policy to reduce future COVID-19 transmission in India.  
\end{abstract}

\begin{keyword}
	COVID-19; Mathematical model; Lockdown effect; Outbreak containment policy
\end{keyword}

\maketitle

\section{Introduction}
As of May 15, 2020, 4593395 cases and 306376 deaths from 2019 novel coronavirus disease (COVID-19), caused by severe acute respiratory syndrome corona­virus 2 (SARS-CoV-2), were recorded worldwide \cite{csse2020}. Coronaviruses are enveloped non-segmented positive-sense RNA viruses that belong to the Coronaviridae family and the order Nidovirales, and are widely distributed among humans and other mammals \cite{richman2016clinical}. The novel coronavirus, COVID-19 started in mainland China, with a geographical emphasis at Wuhan, the capital city of Hubei province \cite{wang2020novel} and has widely spread all over the world. Many of the initial cases were usually introduced to the wholesale Huanan seafood market, which also traded live animals. Clinical trials of hospitalized patients found that patients exhibit symptoms consistent with viral pneumonia at the onset of COVID-19, most commonly fever, cough, sore throat and fatigue \cite{cdcgov2020}. Some patients reported changes in their ground-glass lungs; normal or lower than average white lymphocyte blood cell counts and platelet counts; hypoxemia; and deranged liver and kidney function. Most were said to be geographically related to the wholesale market of Huanan seafood \cite{chinadaily2019}. Severe outbreaks occur in USA ($1457593$ cases), Spain ($272646$ cases), Russia (262843 cases), UK (236711 cases), Italy (223885) and so many countries and the disease continues to spread globally. This has been declared a pandemic by the World Health Organization. It is the third zoonotic human coronavirus that has arisen in the present century, after the $2002$ severe acute respiratory syndrome coronavirus (SARS-CoV), which spread to 37 countries and the 2012 Middle East respiratory syndrome coronavirus (MERS-CoV), which spread to 27 countries.

The 2019 pandemic novel coronavirus was first confirmed in India on January 30, 2020, in the state of Kerala. A total of $82087$ confirmed cases and $2648$ deaths in the country have been reported as of May 15, 2020 \cite{indiacovid2020track}. The Indian government has introduced social distance as a precaution to avoid the possibility of a large-scale population movement that can accelerate the spread of the disease. India government implemented a 14-hour voluntary public curfew on 22 March 2020. Furthermore, the Prime Minister of India also ordered a nationwide lockdown at midnight on March 24, 2020 to slow the
spread of COVID-19. Despite no vaccine, social distancing has identified as the most commonly used prevention and control strategy \cite{ferguson2020report}. The purpose of these initiatives is the restriction of social interaction in workplaces, schools, and other public spheres, except for essential public services such as fire, police, hospitals. No doubt the spread of this virus outbreak has seriously disrupted the life, economy and health of citizens. This is a great concern for everyone how long this scenario will last and when the disease will be controlled.

Mathematical modeling based on system of differential equations may provide a comprehensive mechanism for the dynamics of a disease transmission~\cite{sardar2020realistic}. Several modeling studies have already been performed for the COVID-19 outbreak \cite{tang2020updated,quilty2020effectiveness,shen2020modelling,tang2020estimation, wu2020nowcasting}. Based on data collected from December 31, 2019 till January 28, 2020, Wu et al~\cite{wu2020nowcasting} developed a susceptible exposed infectious recovered model (SEIR) to clarify the transmission dynamics and projected national and global spread of disease. They also calculated around 2.68 is the basic reproductive number for COVID-19. Tang et al~\cite{tang2020estimation} proposed a compartmental deterministic model that would combine the clinical development of the disease, the epidemiological status of the patient and the measures for intervention. Researchers found that the amount of control reproduction number  may be as high as 6.47, and that methods of intervention including intensive touch tracing followed by quarantine and isolation would effectively minimize COVID cases \cite{tang2020estimation}.
For the basic reproductive number, Read et al. reported a value of 3.1 based on the data fitting of an SEIR model, using an assumption of Poisson-distributed  daily time increments \cite{read2020novel}. A report by Cambridge University  has indicated that India's countrywide three-week lockdown would not be adequate to prevent a resurgence of the new coronavirus epidemic that could bounce back in months and cause thousands of infections \cite{singh2020age}. They suggested that two or three lockdowns can extend the slowdown longer with five-day breaks in between or a single 49-day lockdown. Data-driven mathematical modeling plays a key role in disease prevention, planning for future outbreaks and determining the effectiveness of control. Several data-driven modeling experiments have been performed in various regions \cite{tang2020estimation,chen2020mathematical}.
Currently, there are very limited works that studied the impact of lockdown on COVID-19 transmission dynamics in India.

In the present manuscript, we proposed a new mathematical model for COVID-19 that incorporates the lockdown effect. We also considered variability in transmission between symptomatic and asymptomatic populations with former being a fast spreader of the disease. Analyzing COVID-19 daily notified cases from five affected states (Maharashtra, Delhi, Tamil Nadu, Gujarat and Punjab) and from overall India, we studied the effect of social distancing measures implemented by the Govt. on notified cases reduction in those regions. We also estimates the basic reproduction numbers ($R_{0}$) for these six locations. Using a post-processing BMA technique, we ensemble our COVID-19 mathematical model with different statistical forecast model to obtain a projection of notified cases in those six locations for the period May 17, 2020 till May 31, 2020. A global sensitivity analysis is carried out to determine the correlation of two epidemiologically measurable parameters on lockdown effect and also on $R_{0}$. Finally to determine the COVID-19 transmission trend during the projection period (May 17, 2020 till May 31, 2020), we estimate the effective reproduction number ($R_{t}$) for the mentioned six locations.

\section{Method}
Based on the development and epidemiological characteristics of COVID-19, a SEIR type model is more appropriate to study the dynamics of this current pandemic \citep{kucharski2020early, peng2020epidemic, sardar2020realistic}. The model we developed in this paper is based on the interaction of seven mutually exclusive sub-classes namely, Susceptible ($S$), Lockdown ($L$), Exposed ($E$), Asymptomatic ($A$), Symptomatic ($I$), Hospitalized ($C$), and Recovered ($R$).  

Susceptible population ($S$) increased due to constant recruitment rate $\Pi_{H}$ and those individuals coming back from lockdown compartment after the lockdown period $ \displaystyle\frac{1}{\omega}$. Population in the susceptible class decreased due to new infection in contact with symptomatic and asymptomatic infected population, natural death and also a fraction of the susceptible individuals become home-quarantine due to lockdown at a rate $l$. We also assumed variability in disease transmission in asymptomatic and symptomatic population with later being a fast spreader of infection with a variability factor $\left( 0 \leq \rho \leq 1\right)$~\citep{gumel2004modelling, mandal2020prudent}.  

Lockdown population ($L$) increased by those susceptible who are home-quarantined during the lockdown period $ \displaystyle\frac{1}{\omega}$, at a rate $l$. Population under lockdown is decreased due to natural death and those individuals who become susceptible again after the lockdown period $ \displaystyle\frac{1}{\omega}$. For simplicity, we assume ideal lockdown scenario \textit{i.e.} all population under lockdown maintained proper social distancing and do not contribute to new infection. 

Population in the exposed compartment ($E$) increased by new infection coming from susceptible compartment. A fraction $\kappa$ of the exposed individuals become symptomatic infected and remaining fraction ($1-\kappa$) become asymptomatic infected after the disease incubation period $\frac{1}{\sigma}$. Exposed population also decreased due to natural death at a rate $\mu$. 

Asymptomatic infected compartment ($A$) increased due to a fraction ($1-\kappa$) of infection coming from exposed compartment. Since, asymptomatic COVID-19 cases are hard to detect therefore, we assume that asymptomatic infection are not notified. Population in this compartment is decreased due to natural recovery and deaths at a rate $\gamma_1$ and $\mu$, respectively. 

Population in the symptomatic infected compartment ($I$) increased due to a fraction $\kappa$ of infection coming from exposed compartment after the incubation period $\frac{1}{\sigma}$. This compartment decreased due to natural recovery at a rate $\gamma_2$, natural death at a rate $\mu$ and those infected population who are notified \& hospitalized at a rate $\tau$. 

Notified \& hospitalized infected population ($C$) increased due to influx of infection coming from symptomatic infected class at a rate $\tau$. This population decreased due to natural death at a rate $\mu$, disease related deaths at a rate $\delta$, and recovery from COVID-19 at a rate $\gamma_3$. We assume that population of this compartment do not mix with the general population in the community \textit{i.e.} this compartment do not contribute in the COVID-19 transmission.  

Finally, recovered population ($R$) increased due to influx of individuals coming from asymptomatic ($A$), symptomatic ($I$), and notified \& hospitalized individuals ($C$) at a rate $\gamma_1$, $\gamma_2$, and $\gamma_3$, respectively. As we are analyzing this study in a shorter time frame therefore, we assume definitive immunity \textit{i.e.} recovered population do not contribute to new COVID-19 infection. Thus recovered population decreased due to natural death at a rate $\mu$.      

Based on the above assumptions the system of equations that represent COVID-19 transmission with and without lockdown are provided below:

\subsection*{\textbf{Model without lock-down}}
\begin{eqnarray}\label{EQ:eqn 2.1}
\displaystyle{ \frac{dS}{dt} } &=& \Pi_{H}-\frac{\beta_1 I S}{\left(N - C\right) }-\frac{\rho \beta_1 A S}{\left(N-C \right) }-\mu S\nonumber\\
\displaystyle{ \frac{dE}{dt} } &=& \frac{\beta_1 I S}{\left(N - C\right) }+\frac{\rho \beta_1 A S}{\left(N-C \right) }-(\mu+\sigma)E,\nonumber \\
\displaystyle{ \frac{dA}{dt} } &=& (1-\kappa)\sigma E-(\gamma_1+\mu)A,\nonumber \\
\displaystyle{ \frac{dI}{dt} } &=&  \kappa \sigma E-(\gamma_2+\tau+\mu)I, \\
\displaystyle{\frac{dC}{dt} } &=& \tau I-(\delta+\gamma_3+\mu)C, \nonumber \\
\displaystyle{ \frac{dR}{dt} } &=& \gamma_1 A+\gamma_2 I+\gamma_3 C-\mu R ,\nonumber
\end{eqnarray}

\subsection*{\textbf{Model with lock-down}}
\begin{eqnarray}\label{EQ:eqn 3.1}
\displaystyle{ \frac{dS}{dt} } &=& \Pi_{H} + \omega L -\frac{\beta_1 I S}{\left( N-L-C\right) }-\frac{\rho \beta_1 A S}{\left( N-L-C \right) }-\mu S-lS\nonumber\\
\displaystyle{ \frac{dL}{dt} } &=& l S - (\mu + \omega) L,\nonumber\\
\displaystyle{ \frac{dE}{dt} } &=& \frac{\beta_1 I S}{\left( N-L-C\right) }+\frac{\rho \beta_1 A S}{\left( N-L-C \right) }-(\mu+\sigma)E,\nonumber \\
\displaystyle{ \frac{dA}{dt} } &=& (1-\kappa)\sigma E-(\gamma_1+\mu)A,\nonumber \\
\displaystyle{ \frac{dI}{dt} } &=&  \kappa \sigma E-(\gamma_2+\tau+\mu)I, \\
\displaystyle{\frac{dC}{dt} } &=& \tau I-(\delta+\gamma_3+\mu)C, \nonumber \\
\displaystyle{ \frac{dR}{dt} } &=& \gamma_1 A+\gamma_2 I+\gamma_3 C-\mu R .\nonumber
\end{eqnarray}A diagram of our model is provided in Fig~\ref{Fig:Flow_India_covid}. Information of our model~ parameters is provided in Table~\ref{tab:mod1}.

\subsection*{\textbf{Mathematical properties of the model}}
We studied the positivity and boundedness of solution of the model \eqref{EQ:eqn 2.1} (see supplementary appendix). The system \eqref{EQ:eqn 2.1} demonstrates two equilibria, that is, the disease-free equlibrium and an unique endemic equilibrium (see supplementary appendix). The disease-free state is locally asymptotically stable whenever the corresponding basic reproduction number ($R_0$) is less than unity (see supplementary appendix). By using a nonlinear Lyapunov function, it is also seen that the disease-free equilibrium is globally asymptotically stable whenever $R_0<1$ (see supplementary appendix). In addition, the model \eqref{EQ:eqn 2.1} has an unique endemic equilibrium if $R_0$ exceeds unity. Furthermore, using the central manifold theory, the local stability of the endemic equilibrium is established whenever $R_0>1$ (see supplementary appendix).

\subsection*{\textbf{Data}}
Daily COVID-19 reported cases from Maharashtra (MH), Delhi (DL), Tamil Nadu (TN), Gujarat (GJ), Punjab (PJ) and whole India (IND) for the time period March 14, 2020 till May 3, 2020 are considered for our study. These five states are deeply affected by current COVID-19 outbreak in India \citep{indiacovid2020track}. Daily COVID-19 notified cases were collected from \citep{indiacovid2020track}. Demographic data of the different locations are taken from \cite{aadhaar20, Nitiayog2020}.

\subsection*{\textbf{Estimation procedure}}
Several important epidemiological parameters (see Table~\ref{tab:mod1}) of our mathematical model~\eqref{EQ:eqn 3.1} are estimated using COVID-19 daily reported cases from the mentioned six locations. Total time duration of lockdown implemented by Govt. is 54 days starting from March 25, 2020 till May 17, 2020. Time-series data of daily COVID-19 cases in our study for the locations MH, DL, TN, GJ, PJ, and IND, respectively contains both with and without lockdown effect. Therefore, a combination of our mathematical models~\eqref{EQ:eqn 2.1}\&~\eqref{EQ:eqn 3.1} (with and without lockdown) are used for calibration. From our models~\eqref{EQ:eqn 2.1}\&~\eqref{EQ:eqn 3.1}, new COVID-19 notified cases during the $i^{th}$ time interval $\left[t_{i}, t_{i}+\Delta t_{i}\right]$ is 

\begin{equation}
\displaystyle H_{i} (\hat{\theta}) = \displaystyle \tau \int_{t_{i}}^{t_{i} + \Delta t_{i}} I(\xi, \hat{\theta}) \hspace{0.2cm} d\xi,\\
\label{EQ:new-cases-from-model}
\end{equation} where, $\Delta t_{i}$ is the time step length and $\hat{\theta}$ is the set of unknown parameters of the models~\eqref{EQ:eqn 2.1}\&~\eqref{EQ:eqn 3.1} that are estimated. Then $K$ observation from the data and from the models~\eqref{EQ:eqn 2.1}\&~\eqref{EQ:eqn 3.1} are $\left\lbrace D_{1}, D_{2},..., D_{K} \right\rbrace$  and $\lbrace H_{1} (\hat{\theta}) , H_{2} (\hat{\theta}),...., H_{K} (\hat{\theta}) \rbrace $, respectively. Therefore, we constructed the sum of squares function~\citep{sardar2013optimal} as:  

\begin{equation}
\displaystyle SS (\hat{\theta}) = \displaystyle \sum_{i = 1}^{K} \left[D_{i} - H_{i} (\hat{\theta})\right]^2,\\
\label{EQ:sum-of-square-function}
\end{equation}

MATLAB based nonlinear least square solver $fmincon$ is used to fit simulated and observed daily COVID-19 notified cases for the mentioned states and the whole country. Delayed Rejection Adaptive Metropolis~\citep{haario2006dram} (DRAM) algorithm is used to sample the 95\% confidence region. An elaboration of this model fitting technique is provided in \citep{sardar2017mathematical}.  

\subsection*{\textbf{Statistical forecast models and the ensemble model}}
COVID-19 mathematical model we developed in this study may be efficient in capturing the transmission dynamics. However, as solution of the mathematical model is always smooth therefore, our model may not be able to replicate the fluctuations occurring in daily time-series data. Moreover, forecast of future COVID-19 cases based on a single mathematical model may not be very reliable approach. For this purpose, we used two statistical forecast models namely, Auto-regressive Integrated Moving Average (ARIMA); and ARMA errors, Trend and Seasonal components (TBATS) respectively. A Hybrid statistical model (HYBRID) based on the combination of ARIMA and TBATS is also used during forecast. Calibration of three statistical forecast models (ARIMA, TBATS and HYBRID) using COVID-19 daily notified cases from  MH, DL, TN, GJ, PJ, and IND, respectively during March 14, 2020 till May 3, 2020, are done using the $R$ package 'forecastHybrid'~\cite{Forecast2020}. Each individuals models (ARIMA and TBATS) are first fitted to the aforesaid time-series data and then we combined each models with weightage based on in sample error to obtain the HYBRID model~\cite{Forecast2020}. Prediction skill of the each three statistical forecast model (ARIMA, TBATS and HYBRID) are tested on the daily COVID-19 notified cases during May 4, 2020 till May 8, 2020 for each of the six locations (see supplementary Table~\ref{Tab:Goodness-of-Fit}). Based on the prediction skill (see supplementary Table~\ref{Tab:Goodness-of-Fit}), the best statistical forecast model is ensemble with our COVID-19 mathematical models~\eqref{EQ:eqn 2.1} and~\eqref{EQ:eqn 3.1}. A post-processing BMA technique based on 'DRAM' algorithm \cite{haario2006dram} is used to determine the weightage (see supplementary Table~\ref{Tab:estimated-weights} and Fig~\ref{Fig:Marginal-distribution-MH} to Fig~\ref{Fig:Marginal-distribution-IND}) to combine the best statistical model with the COVID-19 mathematical models~\eqref{EQ:eqn 2.1} and~\eqref{EQ:eqn 3.1}.     

\subsection*{\textbf{Disease forecasting under different lockdown scenario}}
Govt. have implemented lockdown all over India on March 25, 2020 and it will continue till May 17, 2020. The short and medium scale industries are largely affected by the lockdown~\cite{Economic_Times2, India_Today4}. To partially recover the economy, Govt. of India continuously relaxing the lockdown rules from April 20, 2020 \cite{Economictimes2020a, financialexpress20a, TheHindu20a}. To forecast COVID-19 cases for the period May 17, 2020 till May 31, 2020, for the six locations (MH, DL, TN, GJ, PJ and IND) based on the Govt. strategy, we considered following scenarios:\vspace{0.2cm}

Forecast based on current lockdown rate: We have estimated the average lockdown rate for our COVID-19 mathematical model (see Table~\ref{tab:mod1} and Table~\ref{Tab:estimated-parameters-Table}). Using this lockdown rate and using other parameters (estimated and known) of our mathematical models~\eqref{EQ:eqn 2.1} \&~\eqref{EQ:eqn 3.1}, we forecast COVID-19 notified cases during May 17, 2020 till May 31, 2020 for the locations MH, DL, TN, GJ, PJ, and IND, respectively. Finally, forecast based on our mathematical model is ensemble with the result based on the best statistical forecast model for a location mentioned earlier.\vspace{0.2cm}   

Forecast based on 15\% reduction in current lockdown rate: We followed same procedure as the previous scenario with 15\% decrement in the estimate of lockdown rate~(see Table~\ref{tab:mod1} and Table~\ref{Tab:estimated-parameters-Table}) to obtained the forecast during the mentioned time period.\vspace{0.2cm}

Forecast based on 20\% reduction in current lockdown rate: we followed the same procedure as previous two scenarios with 20\% decrement in the estimate of lockdown rate (see Table~\ref{tab:mod1} and Table~\ref{Tab:estimated-parameters-Table}) to obtained the forecast during the mentioned time period. \vspace{0.2cm}

Forecast based on 30\% reduction in current lockdown rate: we followed the same procedure as previous three scenarios with 30\% decrement in the estimate of lockdown rate (see Table~\ref{tab:mod1} and Table~\ref{Tab:estimated-parameters-Table}) to obtain the forecast during the mentioned time period.\vspace{0.2cm}   

Forecast based on no lockdown: Continue as earlier and assuming lockdown is lifted after May 17, 2020, we forecast COVID-19 notified cases during May 17, 2020 till May 31, 2020, for the six mentioned locations.\vspace{0.2cm} 
 
\subsection*{\textbf{Estimation of the basic and the effective reproduction number}}
Since we assumed that population under the lockdown do not contact with the infection from the community therefore the basic reproduction number ($R_{0}$) \citep{van2002reproduction} for our mathematical model with and without lockdown~(see Fig~\ref{Fig:Flow_India_covid} and supplementary method) are same and its expression is provided below:

\begin{align*}
R_0=\frac{\beta_1\kappa \sigma}{(\mu+\sigma)(\gamma_2+\tau+\mu)}+\frac{\rho \beta_1(1-\kappa)\sigma}{(\mu+\sigma)(\gamma_1+\mu)}.
\end{align*}

The effective reproductive number ($R_{t}$) is defined as the expected number of secondary infection per infectious in a population made up of both susceptible and non-susceptible hosts \citep{rothman2008modern}. If $R_{t} > 1$, the number of new cases will increase, for $R_{t} =1$, the disease become endemic, and when $R_{t} < 1$ there will be a decline in new cases.

Following~\citep{rothman2008modern}, the expression of $R_{t}$ is given as follows:

\begin{align*}
\displaystyle R_t= R_{0} \times \hat{s},  
\end{align*} where, $\hat{s}$ is the fraction of the host population that is susceptible. 

$R_{0}$ can easily be estimated by plugin the sample values of the unknown parameters (see Table \ref{Tab:estimated-parameters-Table}) of the model without lockdown~\eqref{EQ:eqn 2.1} in the expressions of $R_{0}$. 

Following procedure is adapted to estimate $R_{t}$ during May 17, 2020 till May 31, 2020 under two lockdown scenarios:

\begin{description}
\item[$\bullet$] Using current estimate of the lockdown rate and different parameters of our mathematical model~(see Table~\ref{tab:mod1} and Table~\ref{Tab:estimated-parameters-Table}), we estimate $\hat{s}$ and $R_{t}$ during May 17, 2020 till May 31, 2020, for the locations MH, DL, TN, GJ, PJ and IND, respectively. 

\item[$\bullet$] Using different parameters~(see Table~\ref{tab:mod1} and Table \ref{Tab:estimated-parameters-Table}) of our mathematical model without lockdown~\eqref{EQ:eqn 2.1}, we estimated $\hat{s}$ and $R_{t}$ during May 17, 2020 till May 31, 2020 for the mentioned six locations.	
\end{description} 

\subsection*{\textbf{Sensitivity analysis and effective lockdown strategy}}
To determine an effective lockdown policy in those six locations (MH, DL, TN, GJ, PJ and IND) will require some correlation between lockdown effect with some epidemiologically measurable parameters of our mathematical model~(see Fig~\ref{Fig:Flow_India_covid}). There are several important parameters of our mathematical model~(see Table~\ref{tab:mod1}) and among them there are two parameters that are measurable namely, $\kappa$: fraction of new infected that become symptomatic (COVID-19 testing will provide an accurate estimate) and $\tau$: Average notification \& hospitalization rate of symptomatic COVID-19 infection (this parameter is proportional to the number of COVID-19 testing). Lockdown effect is measured as the difference between the total number of cases projected by our ensemble model with and without lockdown. A global sensitivity analysis~\cite{marino2008methodology} is performed to determine the effect of the mentioned two parameters on the lockdown effect and on the basic reproduction number ($R_{0}$). Using Latin Hyper cube sampling (LHS), we draw $1000$ samples for $\kappa$ and $\tau$, respectively from their respective ranges~(see Table~\ref{tab:mod1}). Partial rank correlation and its corresponding $p$-value are examined to determine the relation between two mentioned parameters with the lockdown effect and $R_{0}$, respectively.

\section{Results and discussion}
Three models (mathematical, statistical forecast and ensemble) fitting to daily COVID-19 notified cases during March 14, 2020 till May 3, 2020, for Maharashtra (MH), Delhi (DL), Tamil Nadu (TN), Gujarat (GJ), Punjab (PJ), and India (IND) is depicted in Fig~\ref{Fig:Model-fitting}. Among the three statistical forecast models (ARIMA, TBATS and HYBRID), the ARIMA model performed better on the test prediction data (May 4, 2020 till May 8, 2020) of DL and GJ (see supplementary Table~\ref{Tab:Goodness-of-Fit}). Whereas, the TBATS model provide better result in compare to other two models for the test prediction data (May 4, 2020 till May 8, 2020) of PJ (see supplementary Table~\ref{Tab:Goodness-of-Fit}). For the remaining three locations (MH, TN and IND), the HYBRID model provide the best result (see supplementary Table~\ref{Tab:Goodness-of-Fit}) on the test prediction data (May 4, 2020 till May 8, 2020). The ensemble model, which is a combination of our COVID-19 mathematical models~\eqref{EQ:eqn 2.1} \&~\eqref{EQ:eqn 3.1} and the best statistical forecast model (region specific), is performed well in capturing COVID-19 daily time-series data trend in all the six locations. Posterior distribution of the weights at which we combine our COVID-19 mathematical model with the best statistical forecast model in the six different locations are provided in supplementary method (see Fig~\ref{Fig:Marginal-distribution-MH} to Fig~\ref{Fig:Marginal-distribution-IND} and Table~\ref{Tab:estimated-weights}).  
  
In MH, DL, and GJ, the estimate of the symptomatic influx fraction ($\kappa$) suggest that low percentage (about 11\% to 20\%) of symptomatic infected in the population (see Table~\ref{Tab:estimated-parameters-Table}). However, in TN and PJ, relatively higher percentage (about 82\% to 88\%) of symptomatic infection is found (see Table~\ref{Tab:estimated-parameters-Table}). In overall India, our estimate shows that currently about 62\% of new infection are symptomatic (see Table~\ref{Tab:estimated-parameters-Table}). Except for GJ, in other five locations, estimate of the transmission rate ($\beta_{1}$) are found to be in same scale (see Table~\ref{Tab:estimated-parameters-Table}). Relatively higher value of $\beta_{1}$ is found in Gujrat (see Table~\ref{Tab:estimated-parameters-Table}). Low value of the transmission variability factor ($\rho$) indicates most of the community infection in GJ are due to contact with the symptomatic infected population. As in GJ, the value of $\kappa$ is found to be small (about 11\%) therefore, relatively smaller symptomatic population producing most of the infection in GJ. This indicate that there may be a possibility of existence of super-spreaders among the symptomatic infected in GJ. This observation is agree with a recent survey result in GJ~\cite{NDTVsuperspreaders}. Except for PJ, in other five locations, the estimates of $\rho$ (below 50\%) are found to be low (see Table~\ref{Tab:estimated-parameters-Table}). This indicates small contribution of the asymptomatic infected population towards the new infection produced in MH, DL, TN, GJ and IND, respectively. Estimate of the lockdown rate in the five states (MH, DL, TN, GJ and PJ) suggest that around 50\% to 88\% of the total susceptible population are successfully home quarantined during the lockdown period (see Table~\ref{Tab:estimated-parameters-Table}). Thus, lockdown is overall successful in those five states. However, this is not the case for overall India, our estimate suggest that about 11\% of the total susceptible population in India maintained proper social distancing during the lockdown period (see Table~\ref{Tab:estimated-parameters-Table}).          

Our estimate of the basic reproduction number ($R_{0}$)~(see Table~\ref{Tab:estimated-R0-Table}), in the six locations found to be in good agreement of the world-wide estimate provided by WHO~\cite{liu2020reproductive}. We performed a global sensitivity analysis of two epidemiologically measurable parameters of our mathematical model~(see Fig~\ref{Fig:Flow_India_covid}) namely $\kappa$: fraction of new infected that become symptomatic and $\tau$: Average notification \& hospitalization rate of symptomatic COVID-19 infection, on $R_{0}$. Partial rank correlation and its corresponding $p$-value (see Fig~\ref{Fig:sensitivity-analysis-R0}) suggest that $\tau$ has a negative correlation on $R_{0}$. Thus, more testing will isolate more infection from the community and therefore may reduce the COVID-19 community transmission. Furthermore, high positive correlation of $\kappa$ with $R_{0}$ (see Fig~\ref{Fig:sensitivity-analysis-R0}) indicates the possibility of high COVID-19 transmission in those areas where population have higher percentage of symptomatic infection.    

Ensemble model forecast of notified COVID-19 cases between May 17, 2020 till May 31, 2020 (see Table~\ref{Tab:cases-preiction-Table}, Fig.~\ref{Fig:Prediction-cases-India}, and Fig~\ref{Fig:Prediction-cases-MH} to Fig~\ref{Fig:Marginal-distribution-PJ} in supplementary appendix) indicate that in the coming few days, a high increment in the COVID-19 notified cases may be observed in MH, DL, TN, GJ, PJ, and IND. Furthermore, our ensemble model prediction during the mentioned period suggest that around 117645 to 128379 cases may occurred in overall India (see Table~\ref{Tab:cases-preiction-Table}). These numbers are much higher than the total cumulative cases between March 2, 2020 till May 15, 2020, in whole India.  
  
A global sensitivity analysis of $\kappa$ and $\tau$ on the lockdown effect suggest that both of these parameters have high positive correlation with the lockdown effect in all the six locations (see Fig.~\ref{Fig:Sensitivity-lockdown-effect}). Therefore, lockdown will be effective in those region where higher percentage of symptomatic infection is found in the population and also larger COVID-19 mass testing will be required to isolate the cases. 

To measure the COVID-19 transmission trend during May 17, 2020 till May 31, 2020, we estimated the effective reproduction number ($R_{t}$) during the mentioned period for MH, DL, TN, GJ, PJ and IND, respectively (see Fig.~\ref{Fig:Effective-reproduction-number}). Our result suggest that, a decreasing trend in new notified COVID-19 cases ($R_{t} < 1$) may be seen after May 31, 2020 if current lockdown measures (see Table~\ref{Tab:estimated-parameters-Table}) are maintained in DL, TN and PJ, respectively. Furthermore, if social distancing measures are removed after May 17, 2020, we may see a rise in the daily COVID-19 cases in all of the six locations (see Fig.~\ref{Fig:Effective-reproduction-number}).

\section{Conclusion}

Up to May 15, 2020, total number of reported COVID-19 cases and deaths in India are \textbf{81794} and \textbf{2649}, respectively \citep{indiacovid2020track}. This tally rises with few thousand new notified cases every day reported from different locations in India \citep{indiacovid2020track}. Currently, there is no treatment or vaccine available for COVID-19. Therefore, only measure to control the outbreak may be home quarantined (lockdown) a larger percentage of susceptible population. However, this disease control strategy may have some negative impact on the economy. Therefore, it is utmost important to determine an effective lockdown policy that may reduce COVID-19 transmission in the community as well as save the Indian economy from drowning. This policy may be found by studying the dynamics and prediction of a mechanistic mathematical model for COVID-19 transmission and testing the results in real life situation.     

In this present study, we consider a new mathematical model on COVID-19 transmission that incorporates the lockdown effect (see Fig~\ref{Fig:Flow_India_covid}). In our models~\eqref{EQ:eqn 2.1} \&~\eqref{EQ:eqn 3.1}, we also considered transmission variability between symptomatic and asymptomatic population with former being a fast spreader of the disease. Using daily time-series data of notified COVID-19 cases from five states (Maharashtra, Delhi, Tamil Nadu, Gujarat and Punjab) and overall India, we studied the effect of lockdown measures on the reduction of notified cases in those regions. Our result suggest that lockdown will be effective in those locations where higher percentage of symptomatic infection exist in the population. Furthermore, a large scale COVID-19 mass testing is required to reduce community infection in those locations. Using a post-processing BMA technique, we ensemble the prediction of our mathematical model with the results obtained from different statistical forecast model. Our ensemble model forecast of COVID-19 daily notified cases during May 17, 2020 till May 31, 2020, suggested a very high rise in the COVID-19 notified cases in the mentioned time duration in most of the locations. Furthermore, estimation of the effective reproduction number ($R_{t}$) during the mentioned time duration indicates if the lockdown measures are completely removed after May 17, 2020, in those locations, a high spike in COVID-19 notified cases may be seen during the mentioned forecasting period. We provide a suggestion for the Indian Govt. and policy makers to acquire the following steps for effective containment of COVID-19 transmission: 

\begin{enumerate}
	\item Perform a survey to find the percentage of symptomatic infection in different states and regions.
	
	\item Focus implementing extensive lockdown in those locations only where the percentage of symptomatic infection is high.
	
	\item Provide relaxation in lockdown in other locations for some time. This process will increase the percentage of symptomatic infection.
	
	\item Repeat step-2, when a region has a sufficient percentage of symptomatic infection.
	
\end{enumerate} 

There are some drawback in our study and may be modified in future. We assume that lockdown population ($L$) and notified \& hospitalized infection ($C$) do not mix with the general population in the community. However, there are numerous evidences where disease transmitted from the hospital and from home confined individuals~\cite{NDTVhospital20}. We shall leave these challenges for our future objectives.

\section*{Conflict of interests}
The authors declare that they have no conflicts of interest.

\section*{Acknowledgments}
The authors are grateful to editor-in-chief, handling editor and learned reviewers for their comments and suggestions on the earlier version of this manuscript. The comments immensely improve the standard of this article.\\

Dr. Tridip Sardar acknowledges the Science \& Engineering Research Board (SERB) major project grant  (File No: EEQ/2019/000008 dt. 4/11/2019), Government of India.\\

Sk Shahid Nadim receives funding as senior research fellowship from Council of Scientific \& Industrial Research (Grant No: 09/093(0172)/2016/EMR-I), Government of India, New Delhi.\\ 

The Funder had no role in study design, data collection and analysis, decision to publish, or preparation of the manuscript.

\clearpage
\bibliographystyle{ieeetr}
\biboptions{square}
\bibliography{covid19_india}

\clearpage
\begin{center}
	\section*{\Large{\underline{Figures}}}
\end{center}

\begin{figure}[ht]
	\captionsetup{width=1.1\textwidth}
		\includegraphics[width=1.0\textwidth]{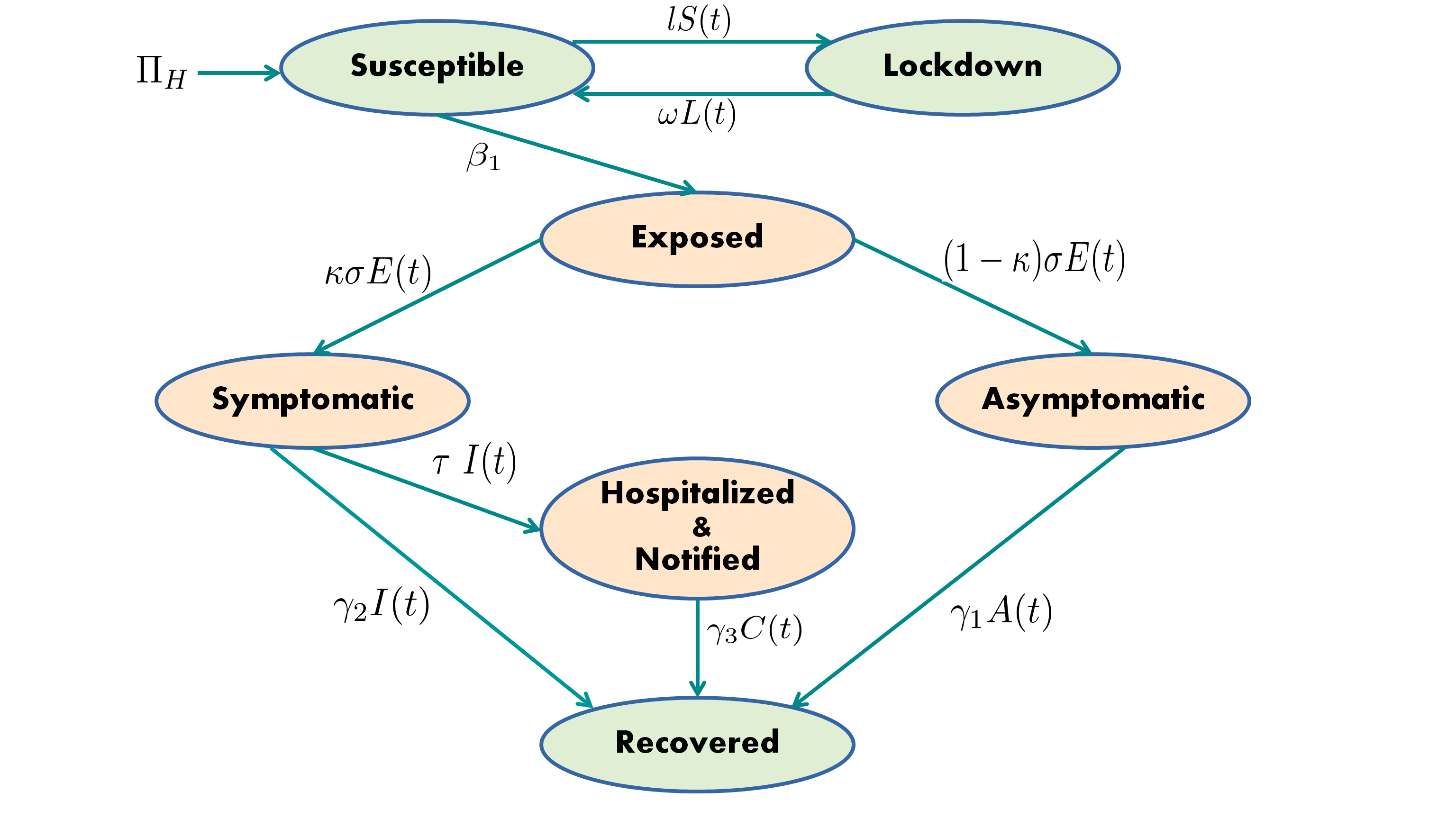}	
	\caption{Flow diagram of the mathematical model with lockdown. Epidemiological information of different parameters shown in the above figure are provided in Table~\ref{tab:mod1}.}
	\label{Fig:Flow_India_covid}
\end{figure}

\begin{figure}[ht]
	\captionsetup{width=1.1\textwidth}
	\includegraphics[width=1.15\textwidth]{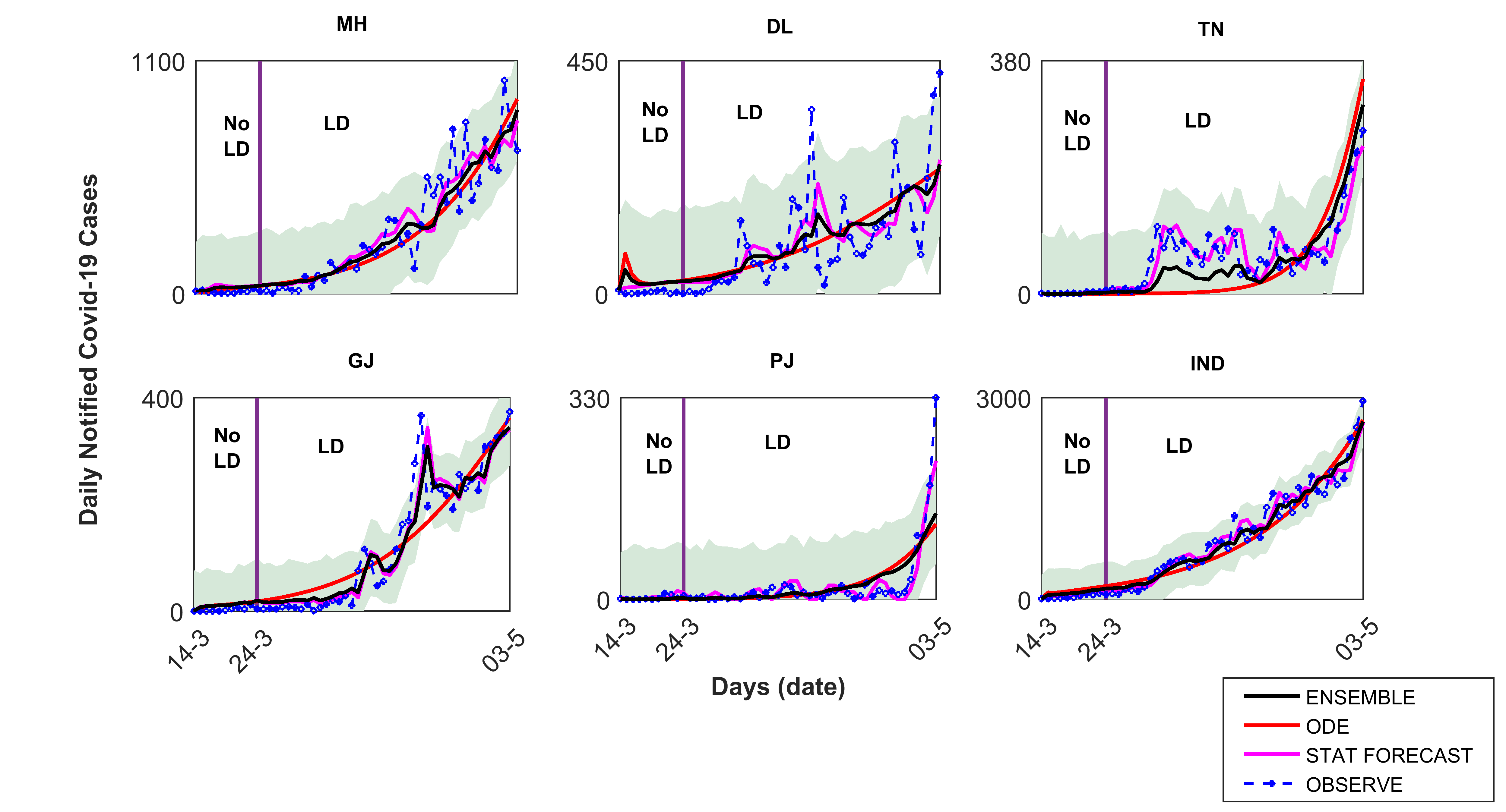}
	
	\caption{Three models (mechanistic mathematical model, statistical forecast model and ensemble model) fitting to daily notified COVID-19 cases from five different states and overall India for the period March 14, 2020 till May 3, 2020. Respective subscript are MH: Maharashtra, DL: Delhi, TN: Tamil Nadu, GJ: Gujarat, PJ: Punjab, and IND: India. Here \textbf{LD} denotes lockdown period and \textbf{No LD} indicate the period before lockdown implementation. Statistical forecast model for different locations are ARIMA (DL and GJ), TBATS (PJ) and HYBRID (MH, TN and IND), respectively. Lockdown effect is only considered for the mechanistic mathematical model and consequently in the ensemble model. Shaded area indicate the 95\% confidence region.} 
	\label{Fig:Model-fitting}
\end{figure}

\begin{figure}[ht]
	\captionsetup{width=1.1\textwidth}
	{\includegraphics[width=1\textwidth]{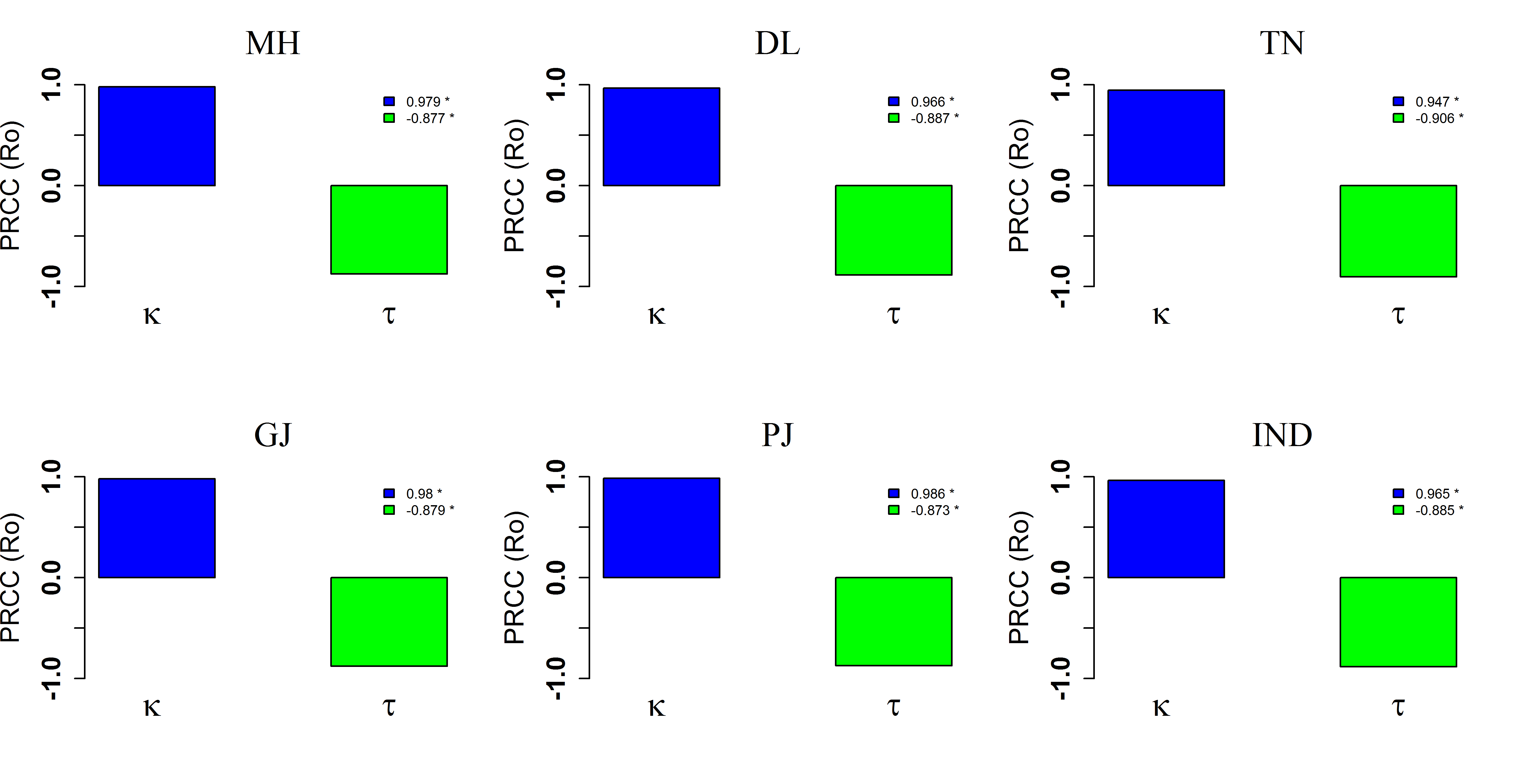}}
	\caption{Global sensitivity analysis of two epidemiologically measurable parameters namely $\kappa$: fraction of new infected that become symptomatic and $\tau$: average notification \& hospitalization rate of symptomatic COVID-19 infection, on $R_{0}$. The subscripts MH, DL, TN, GJ, PJ and IND, respectively are same as Fig~\ref{Fig:Model-fitting}.}
	\label{Fig:sensitivity-analysis-R0}
\end{figure}
\begin{figure}[ht]
	\captionsetup{width=1.1\textwidth}
	\includegraphics[width=1\textwidth]{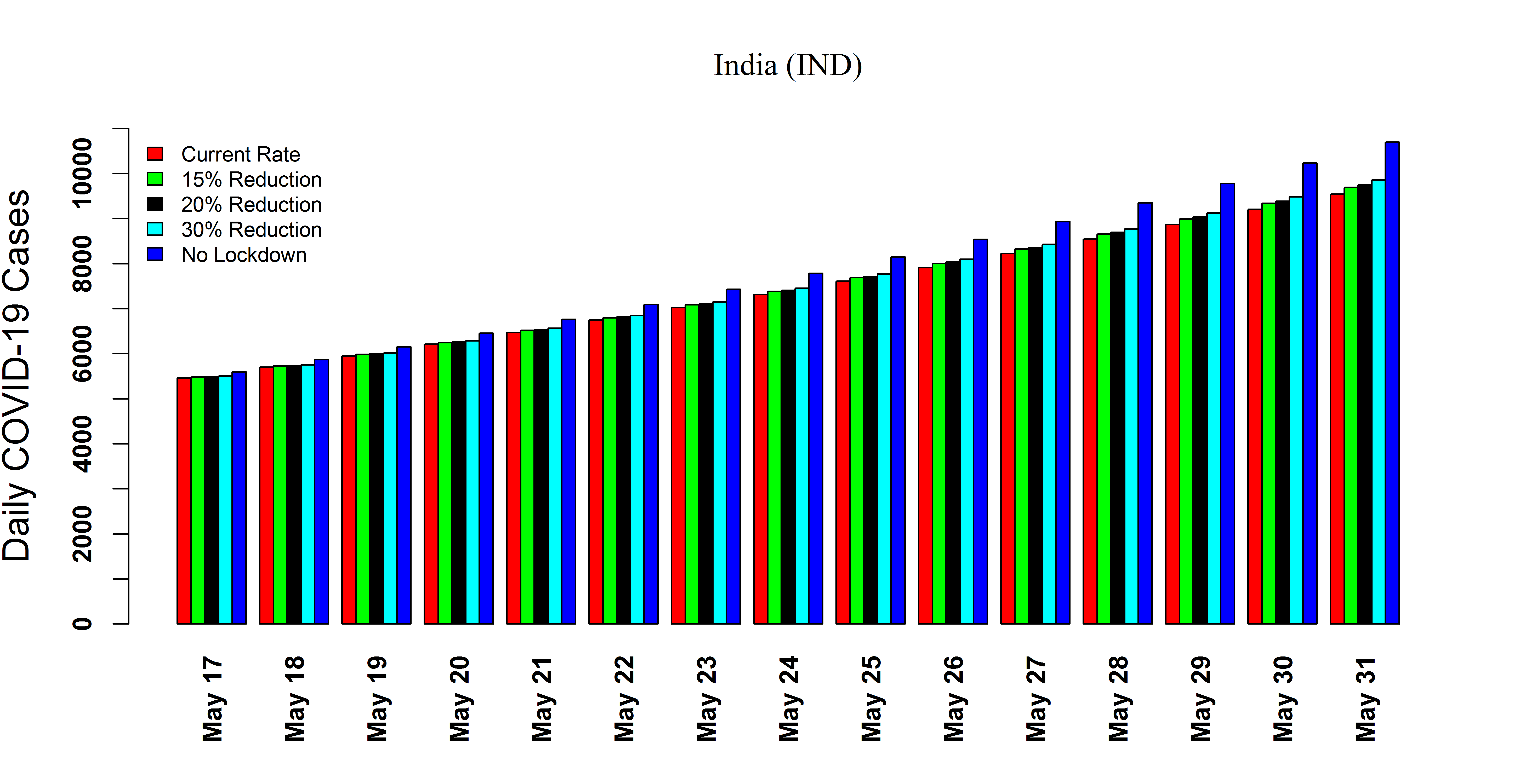}
	\caption{Ensemble model forecast for the daily notified COVID-19 cases in India during May 17, 2020 till May 31, 2020, under five different social distancing measure. Various legends are \textbf{Current Rate}: daily notified case projection using the estimated value of the lockdown rate~(see Table~\ref{Tab:estimated-parameters-Table}), \textbf{15\% Reduction}: daily notified case projection using 15\% reduction in the estimated value of the lockdown rate~(see Table~\ref{Tab:estimated-parameters-Table}), \textbf{20\% Reduction}: daily notified case projection using 20\% reduction in the estimated value of the lockdown rate~(see Table~\ref{Tab:estimated-parameters-Table}), \textbf{30\% Reduction}: daily notified case projection using 30\% reduction in the estimated value of the lockdown rate~(see Table~\ref{Tab:estimated-parameters-Table}), and \textbf{No lockdown}: daily notified case projection based on no lockdown scenario, respectively. }
	\label{Fig:Prediction-cases-India}
\end{figure}

\begin{figure}[ht]
	\captionsetup{width=1.1\textwidth}
	\includegraphics[width=1\textwidth]{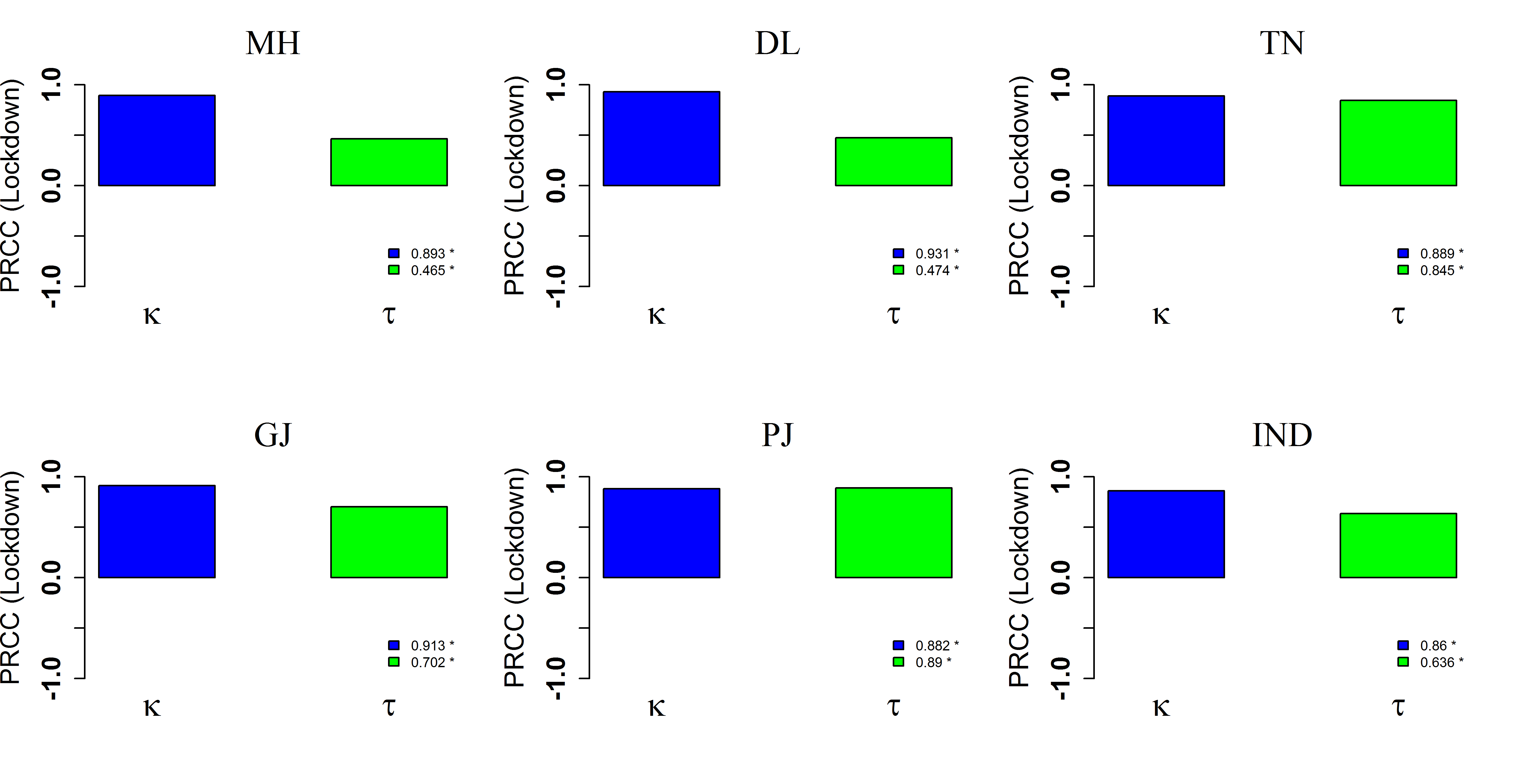}
	\caption{Global sensitivity analysis of two epidemiologically measurable parameters namely $\kappa$: fraction of new infected that become symptomatic and $\tau$: average notification \& hospitalization rate of symptomatic COVID-19 infection, on the lockdown effect. The effect of lockdown is measured as the difference between the total number of cases projected by our ensemble model with and without lockdown scenario during May 17, 2020 till May 31, 2020. The subscripts MH, DL, TN, GJ, PJ and IND, respectively are same as Fig~\ref{Fig:Model-fitting}.}
	\label{Fig:Sensitivity-lockdown-effect}
\end{figure}

\begin{figure}[ht]
	\captionsetup{width=1.1\textwidth}
	\includegraphics[width=1\textwidth]{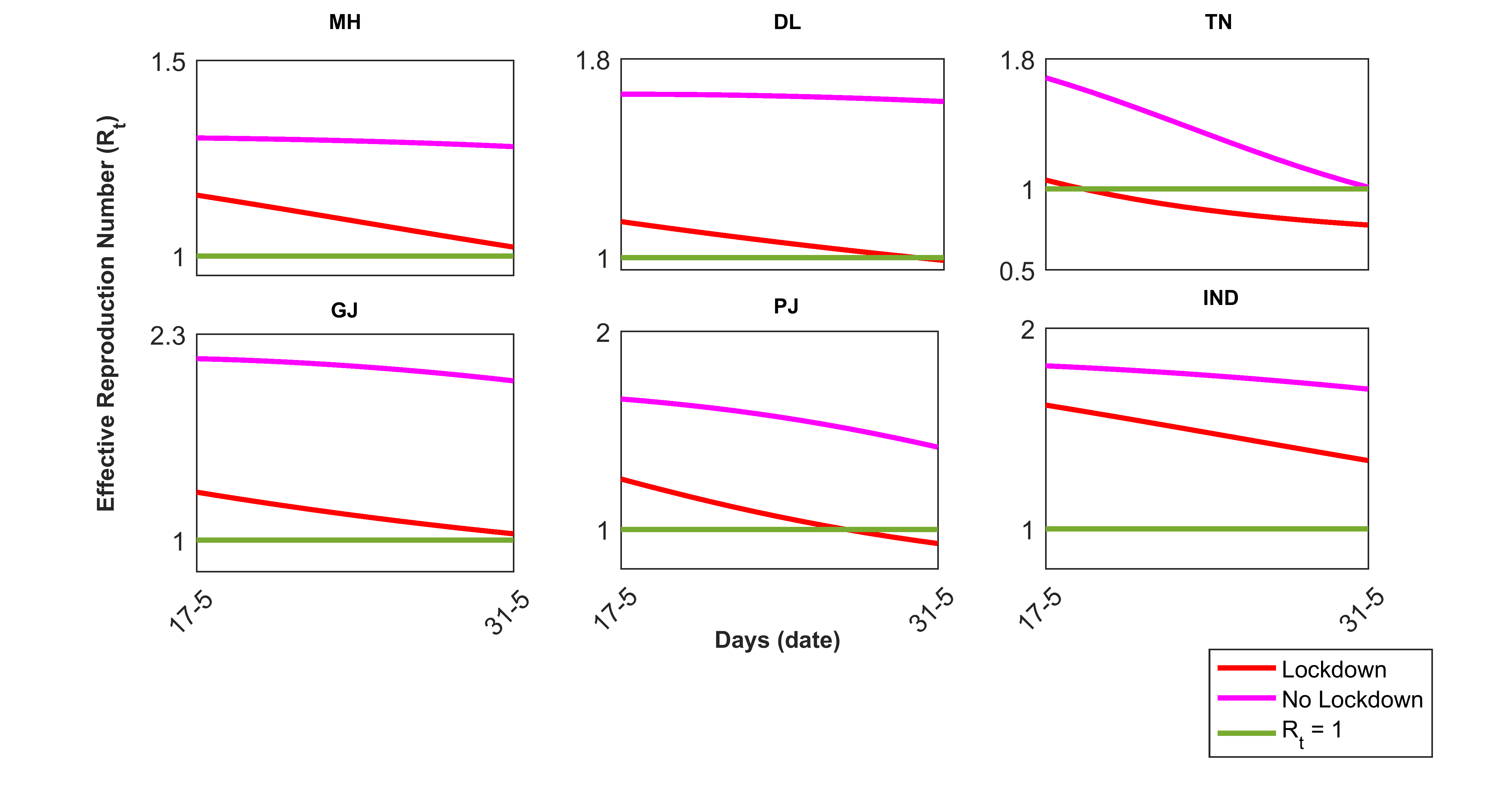}
	\caption{Effective reproduction number ($R_{t}$) for the period May 17, 2020 till May 31, 2020, in MH, DL, TN, GJ, PJ and IND, respectively. The subscripts MH, DL, TN, GJ, PJ and IND, respectively are same as Fig~\ref{Fig:Model-fitting}.}
	\label{Fig:Effective-reproduction-number}
\end{figure}

\clearpage
\begin{center}
\section*{\Large{\underline{Tables}}}
\end{center}

\begin{table}[ht]
	\captionsetup{font=normalsize}
	\captionsetup{width=1.1\textwidth}
	\begin{center}
		\caption{Parameters with their respective epidemiological explanation for the mechanistic mathematical model~(see Fig~\ref{Fig:Flow_India_covid}) for COVID-19.}
		\label{tab:mod1}
		\begin{tabular}{|c| p{8.2cm} p{3.4cm} p{2cm}}
			\hline\\
			\textbf{Parameters} & \textbf{Epidemiological Meaning} &  \textbf{Ranges} & \textbf{Reference}\\\hline\\
			 $N$  & Total population & Varies over different region & \cite{aadhaar20}\\
			$\Pi_{H}$ $=$ $\displaystyle \mu \times N$ & Average recruitment rate & Varies over different region & --\\
			$\displaystyle\frac{1}{\mu}$ & Average life expectancy at birth & Varies over states & \cite{Nitiayog2020}\\\\
			$\beta_1$ & Transmission rate of symptomatic infected & (0 - 200) $day^{-1}$  & Estimated \\\\
			$\rho$ & Reduction in COVID-19 transmission for Asymptomatic infected & 0 - 1 & Estimated \\\\
			$\displaystyle\frac{1}{\sigma}$ & COVID-19 incubation period & (2 - 14) $days$ &  Estimated \\\\
			$\displaystyle \kappa$ & Fraction of exposed population that become
			symptomatic infected & 0 - 1 & Estimated\\\\
			$\gamma_1$ & Recovery rate for asymptomatic
			infected & ($\gamma_{3}$ - 1) $day^{-1}$& Estimated\\\\
			$\gamma_2$ &  Recovery rate for symptomatic
			infected & ($\gamma_{3}$ - 1) $day^{-1}$ & Estimated \\\\
			$\tau$ & Average notification \& hospitalization rate of symptomatic infection 
			&  (0 - 1) $day^{-1}$ & Estimated\\\\
			$\delta$ & Average case fatality rate & Varies over states & \cite{indiacovid2020track}\\\\
			$\gamma_3$ & Average recovery rate for COVID-19 hospitalized \& notified infection & Varies over states & \cite{indiacovid2020track}\\\\
			$l$ & Average lockdown rate & (0 - 1) $day^{-1}$ & Estimated\\\\
			$\frac{1}{\omega}$ & lockdown period & 54 days & \cite{TimesofIndia54days}\\
			\hline
		\end{tabular}
	\end{center}
\end{table}

\begin{table}[ht]
	\captionsetup{font=normalsize}
	\captionsetup{width=1.1\textwidth}
	\tabcolsep 4.5pt
	\footnotesize	
	\centering
	\caption{Estimated parameter values of the mechanistic mathematical model~(see Fig~\ref{Fig:Flow_India_covid}). Respective row subscripts are same as Fig.~\ref{Fig:Model-fitting}. All data are given in the format~\textbf{Estimate (95\% CI)}.}\vspace{0.3cm}
	\begin{tabular}{p{2cm}|ccccccccc} \hline\\
		\textbf{Location} &  $\boldsymbol{\beta_{1}}$ & $\boldsymbol{\rho}$  & $\boldsymbol{\sigma}$  & $\boldsymbol{\kappa}$ & $\boldsymbol{\gamma_1}$ &  $\boldsymbol{\gamma_2}$  & $\boldsymbol{\tau \times 100}$ & $\boldsymbol{l\times 100}$ \\ \hline\\
		\textbf{MH}& \footnotesize{$\substack{5.58  \\  \\ (2.62 - 9.79)}$} & \footnotesize{$\substack{0.015  \\  \\ (0.01 - 0.15)}$} & \footnotesize{$\substack{0.39  \\  \\ (0.07 - 0.45)}$} & \footnotesize{$\substack{0.20  \\  \\ (0.11 - 0.53)}$} & \footnotesize{$\substack{0.63 \\  \\ (0.27 - 0.99)}$} & \footnotesize{$\substack{0.69  \\  \\ (0.22 - 0.98)}$} & \footnotesize{$\substack{24.36  \\  \\ (11.34 - 93.53)}$} & \footnotesize{$\substack{57.84  \\  \\ (11.77 - 86.19)}$}\\\\
		\hline\\
		\textbf{DL}& \footnotesize{$\substack{5.93  \\  \\ (0.78-6.30)}$} & \footnotesize{$\substack{0.13  \\  \\ (0.02 - 0.82)}$} & \footnotesize{$\substack{0.11  \\  \\ (0.07 - 0.47)}$} & \footnotesize{$\substack{0.16  \\  \\ (0.11 - 0.93)}$} & \footnotesize{$\substack{0.99  \\  \\ (0.31 - 0.99)}$} & \footnotesize{$\substack{0.47  \\  \\ (0.33 - 0.98)}$} & \footnotesize{$\substack{37.86  \\  \\ (12.48 - 98.27)}$} & \footnotesize{$\substack{87.03  \\  \\ (11.25 - 87.11)}$}\\\\
		\hline\\
		\textbf{TN}& \footnotesize{$\substack{1.47  \\  \\ (1.23-3.63)}$} & \footnotesize{$\substack{0.37  \\  \\ (0.18 - 0.91)}$} & \footnotesize{$\substack{0.30  \\  \\ (0.15-0.48)}$} & \footnotesize{$\substack{0.88  \\  \\ (0.03 - 0.91)}$} & \footnotesize{$\substack{0.84  \\  \\ (0.29 - 0.98)}$} & \footnotesize{$\substack{0.60  \\  \\ (0.30 - 0.98)}$} & \footnotesize{$\substack{0.56  \\  \\ (0.22 - 9.65)}$} & \footnotesize{$\substack{52.16  \\  \\ (12.38 - 88.92)}$}\\\\
		\hline\\
		\textbf{GJ}& \footnotesize{$\substack{10.73  \\  \\ (5.96 - 20.53)}$} & \footnotesize{$\substack{0.03  \\  \\ (0.002 - 0.15)}$} & \footnotesize{$\substack{0.10  \\  \\ (0.07 - 0.20)}$} & \footnotesize{$\substack{0.11  \\  \\ (0.10 - 0.15)}$} & \footnotesize{$\substack{0.24  \\  \\ (0.21 - 0.98)}$} & \footnotesize{$\substack{0.91  \\  \\ (0.66 - 0.99)}$} & \footnotesize{$\substack{16.79  \\  \\ (10.10 - 23.56)}$} & \footnotesize{$\substack{77.68  \\  \\ (37.15 - 89.36)}$}\\\\
		\hline\\
		\textbf{PJ}& \footnotesize{$\substack{1.63  \\  \\ (1.31-5.35)}$} & \footnotesize{$\substack{0.68  \\  \\ (0.016 - 0.95)}$} & \footnotesize{$\substack{0.26  \\  \\ (0.076-0.45)}$} & \footnotesize{$\substack{0.82  \\  \\ (0.11 - 0.90)}$} & \footnotesize{$\substack{0.43  \\  \\ (0.12 - 0.97)}$} & \footnotesize{$\substack{0.98  \\  \\ (0.24 - 0.99)}$} & \footnotesize{$\substack{2.11  \\  \\ (0.12 - 2.24)}$} & \footnotesize{$\substack{48.18  \\  \\ (12.50 - 87.76)}$}\\\\
		\hline\\
		\textbf{IND}& \footnotesize{$\substack{2.60  \\  \\ (1.04 - 2.92)}$} & \footnotesize{$\substack{0.02  \\  \\ (0.02 - 0.77)}$} & \footnotesize{$\substack{0.08  \\  \\ (0.07 - 0.14)}$} & \footnotesize{$\substack{0.62  \\  \\ (0.30 - 0.98)}$} & \footnotesize{$\substack{0.73  \\  \\ (0.31 - 0.98)}$} & \footnotesize{$\substack{0.79  \\  \\ (0.33 - 0.98)}$} & \footnotesize{$\substack{0.6  \\  \\ (0.2 - 0.9)}$} & \footnotesize{$\substack{10.41  \\  \\ (10.37 - 42.56)}$}\\\\
		\hline
	\end{tabular}
	\label{Tab:estimated-parameters-Table}
\end{table}

\begin{table}[ht]
	\captionsetup{font=normalsize}
	\tabcolsep 16pt	
	\centering
	\caption{Estimated values of the basic reproduction number ($R_{0}$) in MH, DL, TN, GJ, PJ and IND. Respective row subscripts are same as Fig~\ref{Fig:Model-fitting}. All data are given in the format \textbf{Estimate (95\% CI)}.}\vspace{0.3cm}
	
	\begin{tabular}{c|c} \hline\\

		\textbf{Location} & \textbf{Basic Reproduction Number ($\boldsymbol{R_0}$)}\\\hline
		\\
		\textbf{MH}& $\substack{1.3320  \\\\ (1.2369 - 2.3053)}$\\
		\\
		\textbf{DL}& $\substack{1.7569  \\  \\ (1.2409 - 2.2392)}$\\
		\\
		\textbf{TN}& $\substack{2.2169 \\  \\ (1.78 - 5.0374)}$ \\
		\\
		\textbf{GJ}& $\substack{2.3911 \\  \\ (1.6223 - 2.7436)}$ \\
		\\
		\textbf{PJ}& $\substack{1.8077 \\  \\ (1.7342 - 6.3960)}$ \\
		\\
		\textbf{IND}& $\substack{2.0337 \\  \\ (1.6645 - 2.3318)}$\\\\
		\hline
		
	\end{tabular}
	
	\label{Tab:estimated-R0-Table}
\end{table}

\begin{table}[ht]
	\captionsetup{font=normalsize}
	\captionsetup{width=1.1\textwidth}
	\tabcolsep 5pt
	\footnotesize	
	\centering
	\caption{Ensemble model forecast of the total COVID-19 cases during May 17, 2020 till May 31, 2020 in five different states and overall India under five different lockdown scenario. Respective row subscripts are same as Fig.~\ref{Fig:Model-fitting}. Different lockdown scenarios are \textbf{Current Rate}: cumulative case projection using the estimated value of the lockdown rate~(see Table~\ref{Tab:estimated-parameters-Table}), \textbf{15\% Reduction}: cumulative case projection using 15\% reduction in the estimated value of the lockdown rate~(see Table~\ref{Tab:estimated-parameters-Table}), \textbf{20\% Reduction}: cumulative case projection using 20\% reduction in the estimated value of the lockdown rate~(see Table~\ref{Tab:estimated-parameters-Table}), \textbf{30\% Reduction}: cumulative case projection using 30\% reduction in the estimated value of the lockdown rate~(see Table~\ref{Tab:estimated-parameters-Table}), and \textbf{No lockdown}: cumulative case projection based on no lockdown scenario, respectively. All data are provided in the format~\textbf{Estimate (95\% CI)}.}
\begin{tabular}{|p{2cm}|p{2.8cm}p{2.8cm} p{2.8cm} p{2.8cm} p{2.8cm}} \hline\\
		\textbf{Location} & \textbf{Current Rate} &\textbf{15\% Reduction}
		& \textbf{20\% Reduction}
		& \textbf{30\% Reduction} & \textbf{No lockdown} 
		\\\\ \hline\\
		\textbf{MH}& \footnotesize{$\substack{33138  \\  \\ (31910 - 34046) }$} & \footnotesize{$\substack{34602 \\  \\  (34430 - 34730)}$} & \footnotesize{$\substack{35144  \\  \\  (34983 - 35361)}$} & \footnotesize{$\substack{36318  \\  \\  (35532 - 37383)}$} & \footnotesize{$\substack{48635  \\  \\  (41285 - 58578)}$} \\\\
		\hline\\
		\textbf{DL}& \footnotesize{$\substack{6247  \\  \\ (5857 - 6811)}$} & \footnotesize{$\substack{6504  \\  \\ (6253 - 6866)}$} & \footnotesize{$\substack{6601 \\  \\ (6403 - 6887)}$} & \footnotesize{$\substack{6816 \\  \\ (6734 - 6933)}$} & \footnotesize{$\substack{9623  \\  \\ (7535 - 11068)}$} \\\\
		\hline\\
		\textbf{TN}& \footnotesize{$\substack{19344  \\  \\ (19319 - 21503)}$} & \footnotesize{$\substack{20958   \\  \\ (20946 - 21975)}$} & \footnotesize{$\substack{21598 \\  \\ (21591 - 22162)}$} & \footnotesize{$\substack{23076 \\  \\  (22562 - 23107)}$} & \footnotesize{$\substack{60390 \\  \\ (31726 - 62172)}$} \\\\
		\hline\\
		\textbf{GJ}& \footnotesize{$\substack{8227  \\  \\ (8128 - 9039)}$} & \footnotesize{$\substack{8414  \\  \\ (8273 - 9572)}$} & \footnotesize{$\substack{8485 \\  \\ (8328 - 9776)}$} & \footnotesize{$\substack{8646 \\  \\  (8452 - 10233)}$} & \footnotesize{$\substack{11208 \\  \\ (10433 - 17549)}$} \\\\
		\hline\\
		\textbf{PJ}& \footnotesize{$\substack{7427  \\  \\ (7128 - 7517)}$} & \footnotesize{$\substack{8188  \\  \\ (7702 - 8333)}$} & \footnotesize{$\substack{8480 \\  \\ (7924 - 8647)}$} & \footnotesize{$\substack{9139 \\  \\  (8421 - 9354)}$} & \footnotesize{$\substack{19715 \\  \\ (16406 - 20705)}$} \\\\
		\hline\\		
		\textbf{IND}& \footnotesize{$\substack{117645  \\  \\ (110999 - 122572)}$} & \footnotesize{$\substack{119204 \\  \\(111569 - 124863)}$} & \footnotesize{$\substack{119740 \\  \\(111765 - 125650)}$} & \footnotesize{$\substack{120831\\  \\ (112164 - 127255)}$} & \footnotesize{$\substack{128379 \\  \\(114923 - 138352)}$} \\\\
		\hline
	\end{tabular}
	\label{Tab:cases-preiction-Table}
\end{table}

\clearpage
\setcounter{equation}{0}
\setcounter{figure}{0} 
\setcounter{table}{0}
\renewcommand{\theequation}{S-\arabic{equation}}
\renewcommand{\thetable}{S\arabic{table}}
\renewcommand{\thefigure}{S\arabic{figure}}

\section{Supplementary Appendix}	

\subsection*{Positivity and boundedness of the solution for the Model \eqref{EQ:eqn 3.1}}
This subsection is provided to prove the positivity and boundedness of solutions of the system~\eqref{EQ:eqn 3.1} with initial conditions $(S(0), L(0), E(0),A(0),I(0),C(0),R(0))^T\in \mathbb{R}_{+0}^7$. We first state the following lemma.
\begin{lemma} \label{lma1} 
	Suppose $\Omega \subset \mathbb{R} \times \mathbb{C}^n$ is open, $f_i \in C(\Omega, \mathbb{R}), i=1,2,3,...,n$. If $f_i|_{x_i(t)=0,X_t \in \mathbb{C}_{+0}^n}\geq 0$, $X_t=(x_{1t},x_{2t},.....,x_{1n})^T, i=1,2,3,....,n$, then $\mathbb{C}_{+0}^n\lbrace \phi=(\phi_1,.....,\phi_n):\phi \in \mathbb{C}([-\tau,0],\mathbb{R}_{+0}^n)\rbrace$ is the invariant domain of the following equations
	\begin{align*}
	\frac{dx_i(t)}{dt}=f_i(t,X_t), t\geq \sigma, i=1,2,3,...,n.
	\end{align*}
	where $\mathbb{R}_{+0}^n=\lbrace (x_1,....x_n): x_i\geq 0, i=1,....,n \rbrace$ \cite{yang1996permanence}.
\end{lemma}
\begin{proposition}\label{Prop-1}
	The system \eqref{EQ:eqn 3.1} is invariant in $\mathbb{R}_{+0}^7$.
\end{proposition}
\begin{proof}
	By re-writing the system \eqref{EQ:eqn 3.1} we have:
	\begin{eqnarray}
	\frac{dX}{dt} & =B(X(t)), X(0)=X_0\geq 0
	\label{EQ:eqn A.1}
	\end{eqnarray}
	$ B(X(t))=(B_1(X),B_1(X),...,B_7(X))^T$\\
	We note that
	\begin{align*}
	\frac{dS}{dt}|_{S=0}&=\Pi_H + \omega L > 0,
	\frac{dL}{dt}|_{L=0}= l S \geq 0, 
	\frac{dE}{dt}|_{E=0}=\frac{\beta_1 S( I+\rho A)}{N-C}\geq 0,\\  
	\frac{dA}{dt}|_{A=0}&=(1-\kappa)\sigma E\geq 0, 
	\frac{dI}{dt}|_{I=0}=\kappa \sigma E\geq 0, \frac{dC}{dt}|_{C=0}=\tau I\geq 0,\\
	\frac{dR}{dt}|_{R=0}&=\gamma_1 A+\gamma_2 I+\gamma_3 C \geq 0.
	\end{align*}
	Then it follows from the Lemma \ref{lma1} that $\mathbb{R}_{+0}^7$ is an invariant set for the COVID-19 system~\eqref{EQ:eqn 3.1} with lockdown.
\end{proof}
\begin{corollary}
	The system \eqref{EQ:eqn 2.1} is invariant in $\mathbb{R}_{+0}^6$.
\end{corollary}

\begin{proof}
	Proceeding as proposition~\ref*{Prop-1}, we can easily show that $\mathbb{R}_{+0}^6$ is an invariant set for the COVID-19 system~\eqref{EQ:eqn 2.1} without lockdown.
\end{proof}

\begin{lemma}\label{lemma2}
	The system \eqref{EQ:eqn 3.1} is bounded in the region\\ $\Omega=\lbrace(S, L, E,A,I,C,R)\in \mathbb{R}_{+0}^7|S+L+E+A+I+C+R\leq \frac{\Pi_H}{\mu}\rbrace$
\end{lemma}
\begin{proof}
	We have from the system \eqref{EQ:eqn 3.1}:
	\begin{align*}
	&\frac{dN}{dt}=\Pi_H-\mu N-\delta C\leq \Pi_H-\mu N\\
	& \Longrightarrow \lim\limits_{t\rightarrow \infty}sup N(t)\leq \frac{\Pi_H}{\mu}
	\end{align*}
	Hence the system \eqref{EQ:eqn 3.1} is bounded.
\end{proof}

\begin{corollary}
	The system \eqref{EQ:eqn 2.1} is bounded in the region\\ $\Omega^{*}=\lbrace(S, E,A,I,C,R)\in \mathbb{R}_{+0}^6|S+E+A+I+C+R\leq \frac{\Pi_H}{\mu}\rbrace$
\end{corollary}
\begin{proof}
	proceeding same as lemma~\ref*{lemma2}, we can easily show that the system~\eqref{EQ:eqn 2.1} is bounded in $\Omega^{*}$.
\end{proof}
\subsection*{Local stability of disease-free equilibrium (DFE)}
The DFE of the model \eqref{EQ:eqn 3.1} is provided as follows:
\begin{align*}
\varepsilon_0&=(S^0, L^{0}, E^0, A^0, I^0, C^0, R^0)\\
&=\Big(\frac{\Pi_H (\mu + \omega)}{\mu \left(\mu + \omega + l \right) },\frac{\Pi_H l}{\mu \left(\mu + \omega + l \right) }, 0, 0, 0, 0, 0\Big)
\end{align*}
The local stability of $\varepsilon_0$ can be established for the COVID-19 system~\eqref{EQ:eqn 3.1} by using the next generation operator method. Using the notation in \cite{van2002reproduction}, the matrices $F$ for the new infection and $V$ for the transition terms are given, respectively, by
\begin{align*}
F&=\begin{pmatrix}
0 & \rho\beta_1 &  \beta_1 & 0 \\
0&0 & 0 &0  \\
0 &0 & 0 &0  \\
0 &0&0 &0  \\
\end{pmatrix},\\
V&=\begin{pmatrix}
\mu+\sigma &0 & 0 &0 \\
-(1-\kappa)\sigma &\gamma_1+\mu & 0 &0 \\
-\kappa \sigma &0 & \gamma_2+\tau+\mu &0 \\
0 &0&-\tau  &\delta+\gamma_3+\mu\\
\end{pmatrix}.
\end{align*}
It follows that the basic reproduction number \cite{hethcote2000mathematics}, denoted by $R_0=\Phi(FV^{-1})$, where $\Phi$ is the spectral radius, is given by
\begin{align*}
R_0=\frac{\beta_1\kappa \sigma}{(\mu+\sigma)(\gamma_2+\tau+\mu)}+\frac{\rho \beta_1(1-\kappa)\sigma}{(\mu+\sigma)(\gamma_1+\mu)}
\end{align*}
Using Theorem 2 in \cite{van2002reproduction}, the following result is established.
\begin{lemma}	\label{EQ:eqn L.1}
	The DFE, $\varepsilon_0$, of the model \eqref{EQ:eqn 3.1} is locally-asymptotically stable (LAS) if $R_0<1$, and unstable if $R_0>1$.
\end{lemma}
The threshold quantity, $R_0$ is the basic reproduction number of the disease \cite{hethcote2000mathematics,anderson1979population,anderson1992infectious}. This represent the average number of secondary cases generated by a infected person in a fully susceptible population. The epidemiological significance of \ref{EQ:eqn L.1} is that when $R_0$ is less than unity, a low influx of infected individuals into the population will not cause major outbreaks, and the disease would die out in time.

\subsection*{Global stability of DFE}
\begin{theorem}
	The DFE of the model \eqref{EQ:eqn 3.1} is globally asymptotically stable in $\Omega$ whenever $R_0 \leq 1$.
\end{theorem}
\begin{proof}
	Consider the following Lyapunov function
	\begin{align*}
	\mathcal{L}=\Big(\frac{\sigma(\kappa k_2+\rho(1-\kappa)k_3)}{k_1 k_2}\Big )E+\Big(\frac{\rho k_3}{k_2}\Big)A+I
	\end{align*}
	where $k_1=\mu+\sigma$, $k_2=\gamma_1+\mu$ and $k_3=\gamma_2+\tau+\mu$.\\
	We take the Lyapunov derivative with respect to $t$, 
	\begin{align*}
	\dot{\mathcal{L}}&=\Big(\frac{\sigma(\kappa k_2+\rho(1-\kappa)k_3)}{k_1 k_2}\Big)\dot{E}+\Big(\frac{\rho k_3}{k_2}\Big)\dot{A}+\dot{I}\\
	&=\frac{\sigma(\kappa k_2+\rho(1-\kappa)k_3)}{k_1 k_2}\Big[\frac{\beta_1 S (I+\rho A)}{N-L-C}-k_1 E\Big]+\frac{\rho k_3}{k_2}[(1-\kappa)\sigma E-k_2 A]+(\kappa \sigma E-k_3 I)\\
	&\leq \frac{\beta_1\sigma(\kappa k_2+\rho (1-\kappa)k_3)}{k_1 k_2} (I+\rho A)-\frac{\sigma(\kappa k_2+\rho (1-\kappa)k_3)}{k_2} E+\frac{\rho (1-\kappa)k_3\sigma}{k_2}E\\
	&-\rho k_3 A+\kappa \sigma E-k_3 I   \text{ (Since $S \leq N-L-C $ in $\Omega$)}\\
	&= \frac{\beta_1\sigma(\kappa k_2+\rho (1-\kappa)k_3)}{k_1 k_2} (I+\rho A)-\rho k_3 A -k_3 I\\
	&=\frac{\beta_1\sigma(\kappa k_2+\rho (1-\kappa)k_3)}{k_1 k_2 k_3} k_3(I+\rho A)-\rho k_3 A -k_3 I\\
	&\leq k_3(R_0-1)(I+\rho A)\leq 0, \text{  whenever $R_0\leq 1$.}
	\end{align*}
	Since all the variables and parameters of the model \eqref{EQ:eqn 3.1} are non-negative, it follows that $\dot{\mathcal{L}}\leq 0$ for $R_0\leq 1$ with $\dot{\mathcal{L}}=0$ at diseases free equilibrium. Hence, $\mathcal{L}$ is a Lyapunov function on $\Omega$. Therefore, followed by LaSalle’s Invariance Principle \cite{lasalle1976stability}, that
	\begin{equation}\label{EQ:eqn A.4}
	(E(t),A(t), I(t))\rightarrow (0,0,0) \text{ as } t \rightarrow \infty
	\end{equation}
	Since $\lim\limits_{t\rightarrow \infty}sup I(t)=0 $ (from \ref{EQ:eqn A.4}), it follows that, for sufficiently small $\epsilon>0$, there exist constants $B_1>0$ such that $\lim\limits_{t\rightarrow \infty}sup I(t)\leq \epsilon$ for all $t>B_1$.\\
	Hence, it follows from the sixth equation of the model \eqref{EQ:eqn 3.1} that, for $t> B_1$,
	\begin{align*}
	\frac{dC}{dt}\leq \tau \epsilon-k_4 C
	\end{align*}
	Therefore using comparison theorem \cite{smith1995theory}
	\begin{align*}
	C^{\infty}=\lim\limits_{t\rightarrow \infty}sup C(t)\leq \frac{\tau \epsilon}{k_4}
	\end{align*}
	So as $\epsilon \rightarrow 0$,  $C^{\infty}=\lim\limits_{t\rightarrow \infty}sup C(t)\leq0$\\
	Similarly by using $\lim\limits_{t\rightarrow \infty}inf I(t)=0$, it can be shown that
	\begin{align*}
	C_{\infty}=\lim\limits_{t\rightarrow \infty}inf C(t)\geq 0
	\end{align*}\\
	Thus, it follows from above two relations
	\begin{align*}
	C_{\infty} \geq 0 \geq C^{\infty}
	\end{align*}
	Hence $\lim\limits_{t\rightarrow \infty} C(t)= 0$\\
	Similarly, it can be shown that 
	\begin{align*}
	\lim\limits_{t\rightarrow \infty} R(t)= 0, \lim\limits_{t\rightarrow \infty} S(t)= \frac{\Pi_H \left(\mu + \omega\right)}{\mu \left(\mu + \omega + l\right)}, \text{and} \lim\limits_{t\rightarrow \infty} L(t)= \frac{\Pi_H l}{\mu \left(\mu + \omega + l\right)}.
	\end{align*}
	Therefore by combining all above equations, it follows that each solution of the model equations \eqref{EQ:eqn 3.1}, with initial
	conditions $\in \Omega$ , approaches $\varepsilon_0$ as $t\rightarrow \infty $ for $R_0\leq 1$.
\end{proof}

\subsection*{Existence and stability of endemic equilibria}

In this section, the existence of the endemic equilibrium of the model \eqref{EQ:eqn 3.1} is established. Let us denote
\begin{align*}
M_1 & =\frac{\mu+\omega}{\mu + \sigma}, M_2=\frac{\left(1 - \kappa\right) \sigma \left(\mu+\omega\right) }{\left(\mu + \gamma_1\right) \left(\mu + \sigma\right) }, M_3=\frac{\kappa \sigma \left(\mu+\omega\right) }{\left(\mu + \gamma_2 + \tau\right) \left(\mu + \sigma\right)},\\
M_4 &=\frac{\kappa \tau \sigma \left(\mu+\omega\right) }{\left(\mu + \gamma_2 + \tau\right) \left(\mu + \gamma_3 + \delta\right) \left(\mu + \sigma\right)}.
\end{align*}
Let $\varepsilon^*=(S^*, E^*, A^*, I^*, C^*, R^*)$ represents any arbitrary endemic equilibrium point (EEP) of the model \eqref{EQ:eqn 3.1}. Further, define
\begin{align}\label{EQ:eqn A.2}
\lambda^*=\frac{\beta_1(I^*+\rho A^*)}{N^*-L^*-C^*}
\end{align}
It follows, by solving the equations in \eqref{EQ:eqn 3.1} at steady-state, that
\begin{align}\label{EQ:eqn A.3}
S^*&= \frac{\left(\mu+\omega\right) L^{*}}{l}, L^{*} = \frac{\Pi_{H} l}{\lambda^* \left(\mu+\omega\right) + \mu \left(\mu+\omega +l\right)}, E^*=\frac{M_{1} L^{*} \lambda^*}{l},\\\nonumber  A^*&=\frac{M_{2} L^{*} \lambda^*}{l},
I^* =\frac{M_{3} L^{*} \lambda^*}{l}, C^*=\frac{M_{4} L^{*} \lambda^*}{l}\\
R^* &=\frac{\left(\gamma_1 M_2 + \gamma_2 M_3 + \gamma_3 M_4\right) L^{*} \lambda^*}{l \mu}\nonumber
\end{align}
Substituting the expression in \eqref{EQ:eqn A.3} into \eqref{EQ:eqn A.2} shows that the non-zero equilibrium of the model \eqref{EQ:eqn 2.1} satisfy the following linear equation, in terms of $\lambda^*$:
\begin{align}
A \lambda^* = B
\end{align}
where 
\begin{align*}
A&= \mu M_{1} + M_{2} \left(\mu + \gamma_1\right) + M_{3} \left(\mu + \gamma_2\right) + \gamma_3 M_{4} \\
B&=\mu \left(\mu+\omega\right) (R_{0} - 1)
\end{align*}
Since, $M_{1}>0$, $M_{2}>0$, $M_{3}>0$, and $M_{4}>0$ $\implies$ $A>0$, it is clear that the model \eqref{EQ:eqn 3.1} has an unique endemic equilibrium point (EEP) whenever $R_0>1$ and no positive endemic equilibrium point whenever $R_0<1$. This rules out the possibility of the existence of equilibrium other than DFE whenever $R_0<1$. Therefore, we have the following result:

\begin{theorem}
	The model \eqref{EQ:eqn 3.1} has a unique endemic (positive) equilibrium, given by $\varepsilon^*$, whenever $R_0>1$ and has no endemic equilibrium for $R_0\leq 1$.
\end{theorem}
Now we will prove the local stability of endemic equilibrium.
\begin{theorem}
	The endemic equilibrium $\varepsilon^*$ of the COVID-19 system~\eqref{EQ:eqn 3.1} with lockdown is locally asymptotically stable if $R_0>1$.
\end{theorem}
\begin{proof}
	The Jacobian matrix of the system \eqref{EQ:eqn 3.1} $J_{\varepsilon_0}$ at DFE is given by
	\begin{align*}
	J_{\varepsilon_0}={\begin{pmatrix}
		-\left(\mu + l\right) & \omega &0 &-\rho \beta_1 &- \beta_1 & 0 &0  \\
		 l & -\left(\mu + \omega\right) &0 &0 &0 & 0 &0  \\
		0& 0 &-(\mu+\sigma) & \rho \beta_1 &\beta_1 &0&0 \\
		0 & 0 & (1-\kappa)\sigma & -(\mu + \gamma_1)&0 & 0 &0 \\
		0& 0 &\kappa \sigma & 0&-(\mu+\gamma_2+\tau)  &0 &0 \\
		0 &0 &0&0  &\tau&-(\mu+\gamma_3+\delta) &0 \\
		0 &0 &0 &\gamma_1 & \gamma_2&\gamma_3&-\mu \\
		\end{pmatrix}},
	\end{align*}
	
	Here, by taking $\beta_1$ as a bifurcation parameter, we use the central manifold theory method to determine the local stability of the endemic equilibrium \cite{castillo2004dynamical}. Taking $\beta_1$ as the bifurcation parameter and gives critical value of $\beta_1$ at $R_0=1$ is given as
	
	\begin{equation*}
	\beta_1^*=\frac{(\mu+\sigma)(\gamma_1+\mu)(\gamma_2+\tau+\mu)}{[\kappa \sigma(\gamma_1+\mu)+(1-\kappa)\rho \sigma(\gamma_2+\tau+\mu)]}    
	\end{equation*}
	
	The Jacobian of \eqref{EQ:eqn 3.1} at $\beta=\beta_1^*$, denoted by $J_{\varepsilon_0}|_{\beta=\beta_1^*}$ has a right eigenvector (corresponding to the zero eigenvalue) given by
	$w=(w_1, w_2, w_3, w_4, w_5, w_6, w_7)^T$ , where
	\begin{align*}
	w_1 &=-\frac{ \left(\mu+\sigma\right) \left(\mu+\omega\right)  }{\mu \left(\mu+\omega + l\right)}, w_2=-\frac{ \left(\mu+\sigma\right) l }{\mu \left(\mu+\omega + l\right)}, w_3=1, w_4=\frac{\left(1-\kappa\right) \sigma}{\mu + \gamma_1},\\
	w_5& =\frac{\kappa \sigma}{\mu + \gamma_2 + \tau}, w_6 = \frac{\kappa \sigma \tau}{\left( \mu + \gamma_2 + \tau\right) \left(\mu + \gamma_{3} + \delta\right) }\\ 
	w_7&=\frac{\gamma_1(1-\kappa)\sigma}{\mu(\gamma_1+\mu)}+\frac{\gamma_2 \kappa \sigma }{\mu (\gamma_2+\tau+\mu)}
	+\frac{\gamma_3\tau \kappa \sigma}{\mu(\gamma_2+\tau+\mu)(\delta+\gamma_3+\mu)}
	\end{align*}
	Similarly, from  $J_{\varepsilon_0}|_{\beta=\beta_1^*}$, we obtain a left eigenvector $v=(v_1, v_2, v_3, v_4, v_5, v_6, v_7)$ (corresponding to the zero eigenvalue), where
	\begin{align*}
	& v_1=0, v_2=0, v_3 = 1, v_4=\frac{\rho \beta_1^*}{\gamma_1+\mu}, v_5=\frac{\beta_1^*}{\gamma_2+\tau+\mu}, v_6=0, v_7=0.
	\end{align*}
	Selecting the notations $S=x_1$, $L=x_2$, $E=x_3$, $A=x_4$, $I=x_5$, $C=x_6$, $R=x_7$ and $\frac{dx_i}{dt}=f_i$. Now we calculate the following second-order partial derivatives of $f_i$ at the disease-free equilibrium $\varepsilon_0$ and obtain
	\begin{align*}
	&\frac{\partial^2 f_3}{\partial x_4 \partial x_3}=-\frac{\rho \beta_1^{*} \mu \left(\mu + \omega + l \right) }{\Pi_H \left(\mu + \omega \right) } =\frac{\partial^2 f_2}{\partial x_3 \partial x_4},\\ &\frac{\partial^2 f_3}{\partial x_5 \partial x_3}=-\frac{\beta_1^{*} \mu \left(\mu + \omega + l \right) }{\Pi_H \left(\mu + \omega \right) } =\frac{\partial^2 f_3}{\partial x_3 \partial x_5},\\ &\frac{\partial^2 f_3}{\partial x_4^2}=-\frac{2 \rho \beta_1^{*} \mu \left(\mu + \omega + l \right) }{\Pi_H \left(\mu + \omega \right)},\\ &\frac{\partial^2 f_3}{\partial x_5^2}=-\frac{2 \beta_1^{*} \mu \left(\mu + \omega + l \right) }{\Pi_H \left(\mu + \omega \right)},\\ &\frac{\partial^2 f_3}{\partial x_4 \partial x_5}=-\frac{\left(1+\rho\right) \beta_1^{*} \mu \left(\mu + \omega + l \right) }{\Pi_H \left(\mu + \omega \right) } =\frac{\partial^2 f_3}{\partial x_5 \partial x_4},\\ &\frac{\partial^2 f_3}{\partial x_7 \partial x_4}=-\frac{\rho \beta_1^{*} \mu \left(\mu + \omega + l \right) }{\Pi_H \left(\mu + \omega \right) } =\frac{\partial^2 f_3}{\partial x_4 \partial x_7},\\ &\frac{\partial^2 f_3}{\partial x_7 \partial x_5}=-\frac{\beta_1^{*} \mu \left(\mu + \omega + l \right) }{\Pi_H \left(\mu + \omega \right) } =\frac{\partial^2 f_3}{\partial x_5 \partial x_7}.
	\end{align*}
	Now we calculate the coefficients $a$ and $b$ defined in Theorem 4.1 \cite{castillo2004dynamical} of Castillo-Chavez and Song as follow
	\begin{align*}
	a= \sum_{k,i,j=1}^{6} v_k w_i w_j\frac{\partial^2 f_k(0, 0)}{\partial x_i \partial x_j}
	\end{align*}
	and
	\begin{align*}
	b= \sum_{k,i=1}^{6} v_k w_i\frac{\partial^2 f_k(0, 0)}{\partial x_i \partial \beta}
	\end{align*}
	Replacing the values of all the second-order derivatives measured at DFE and $\beta_1=\beta_1^*$, we get
	\begin{align*}
	a&=-2\left(w_{3}+w_{4}+w_{5}+w_{7}\right) \left[\frac{\rho \beta_1^{*} \mu \left(\mu + \omega + l \right) }{\Pi_H \left(\mu + \omega \right) } w_{4} + \frac{\beta_1^{*} \mu \left(\mu + \omega + l \right) }{\Pi_H \left(\mu + \omega \right) } w_{5}\right] <0
	\end{align*}
	and 
	\begin{align*}
	b&= \rho w_4+w_5 \\
	&= \frac{\left(1-\kappa\right) \sigma \rho }{\left(\mu + \gamma_{1}\right)} + \frac{\kappa \sigma }{\left(\mu + \gamma_{2} +\tau\right)}>0.
	\end{align*}
	Since $a<0$ and $b>0$ at $\beta=\beta_1^*$, therefore using the Remark 1 of the Theorem 4.1 stated in \cite{castillo2004dynamical}, a transcritical bifurcation occurs at $R_0=1$ and the unique endemic equilibrium of the COVID-19 system~\eqref{EQ:eqn 3.1} with lockdown is locally asymptotically stable for $R_0>1$.
\end{proof}
\clearpage

\begin{center}
	\section*{\Large{\underline{Figures}}}
\end{center}

\begin{figure}[ht]
	\captionsetup{width=1.1\textwidth}
	\includegraphics[width=1\textwidth]{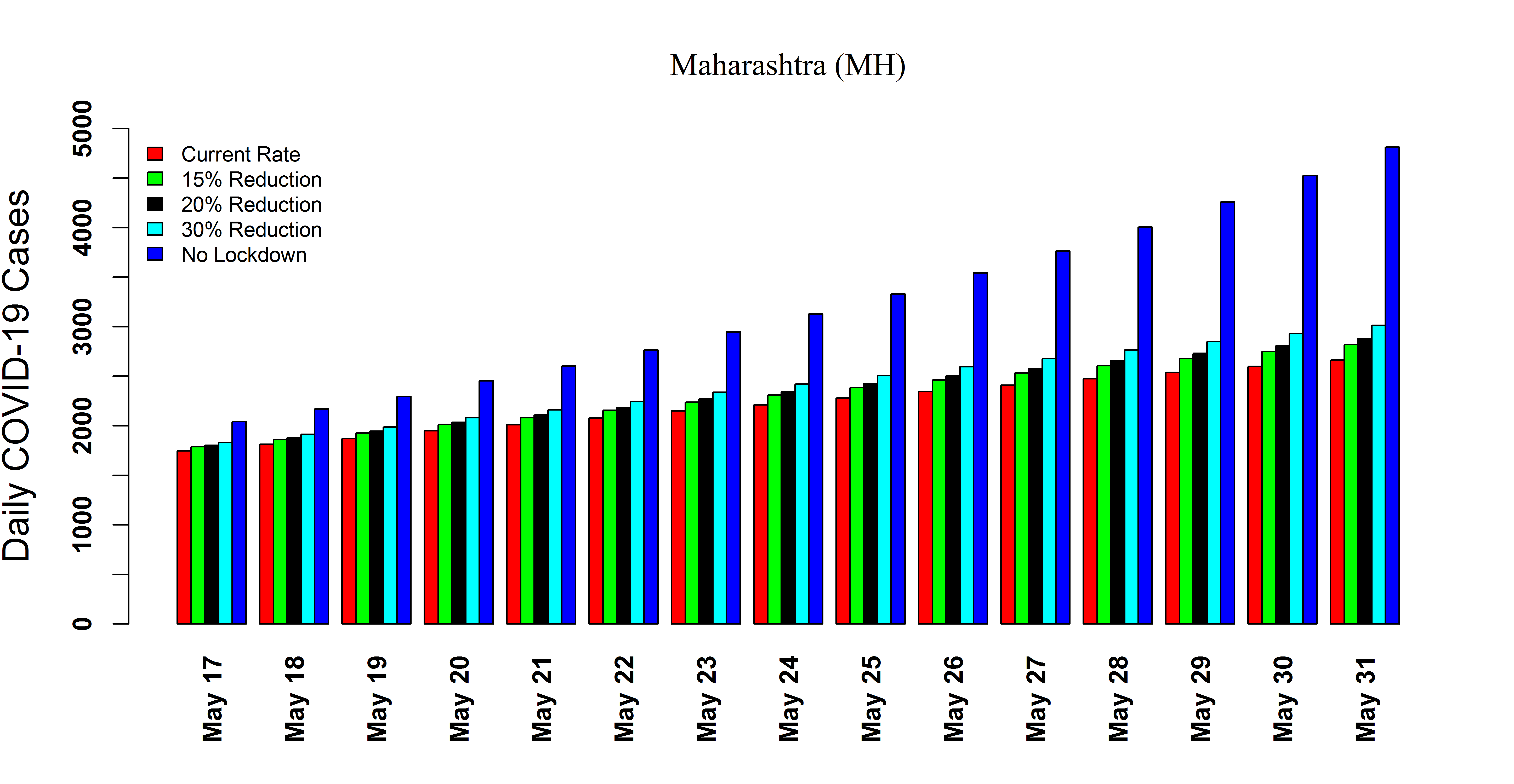}
	\caption{Ensemble model forecast for the daily notified COVID-19 cases in Maharashtra during May 17, 2020 till May 31, 2020, under five different social distancing measure. Various legends are \textbf{Current Rate}: daily notified case projection using the estimated value of the lockdown rate~(see Table~4 main text), \textbf{15\% Reduction}: daily notified case projection using 15\% reduction in the estimated value of the lockdown rate~(see Table~4 main text), \textbf{20\% Reduction}: daily notified case projection using 20\% reduction in the estimated value of the lockdown rate~(see Table~4 main text), \textbf{30\% Reduction}: daily notified case projection using 30\% reduction in the estimated value of the lockdown rate~(see Table~4 main text), and \textbf{No lockdown}: daily notified case projection based on no lockdown scenario, respectively. }
	\label{Fig:Prediction-cases-MH}
\end{figure}

\begin{figure}[ht]
	\captionsetup{width=1.1\textwidth}
	\includegraphics[width=1\textwidth]{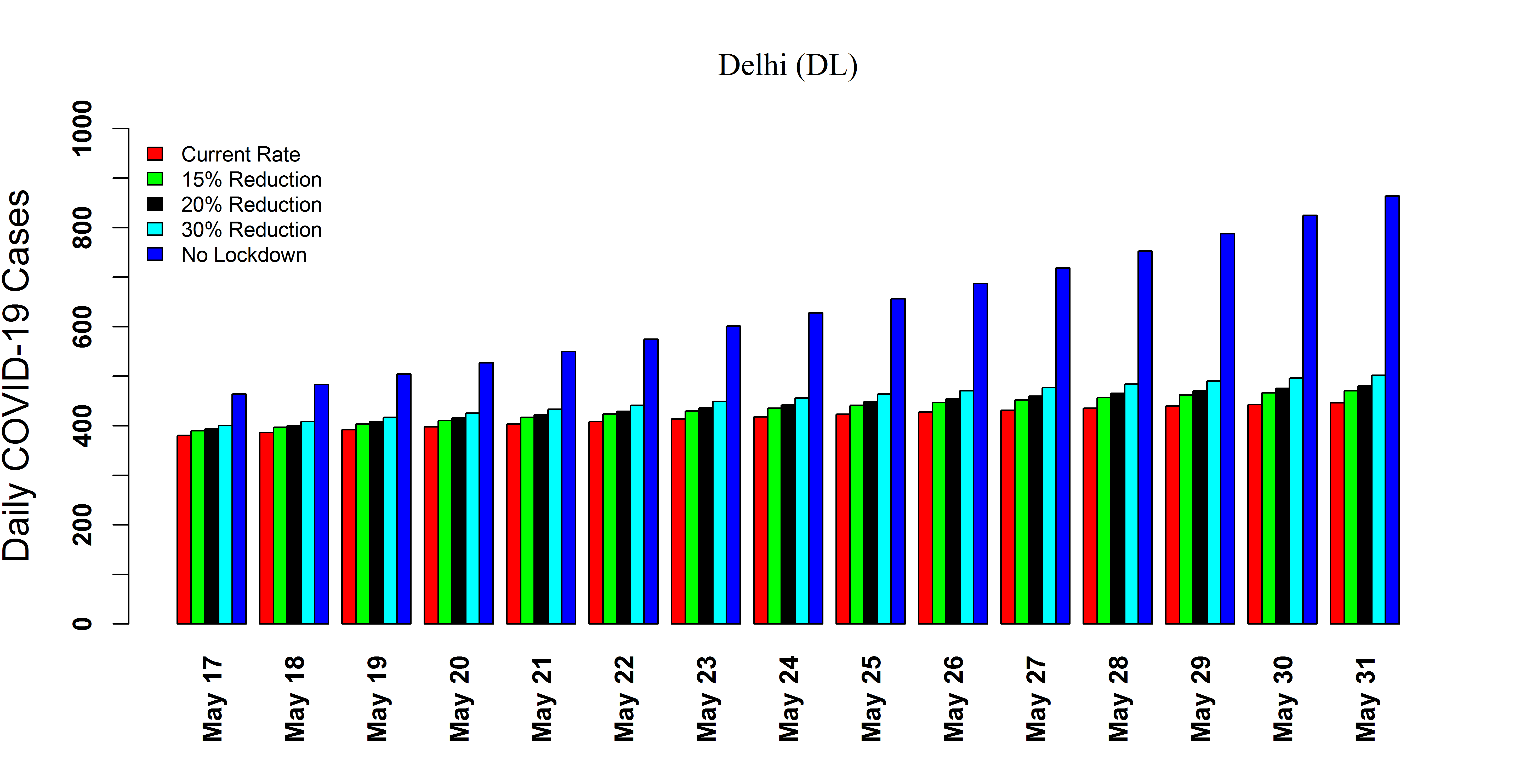}
	\caption{Ensemble model forecast for the daily notified COVID-19 cases in Delhi during May 17, 2020 till May 31, 2020, under five different social distancing measure. Various legends are same as Fig~\ref{Fig:Prediction-cases-MH}.}
	\label{Fig:Prediction-cases-DL}
\end{figure}
\begin{figure}[ht]
	\captionsetup{width=1.1\textwidth}
	\includegraphics[width=1\textwidth]{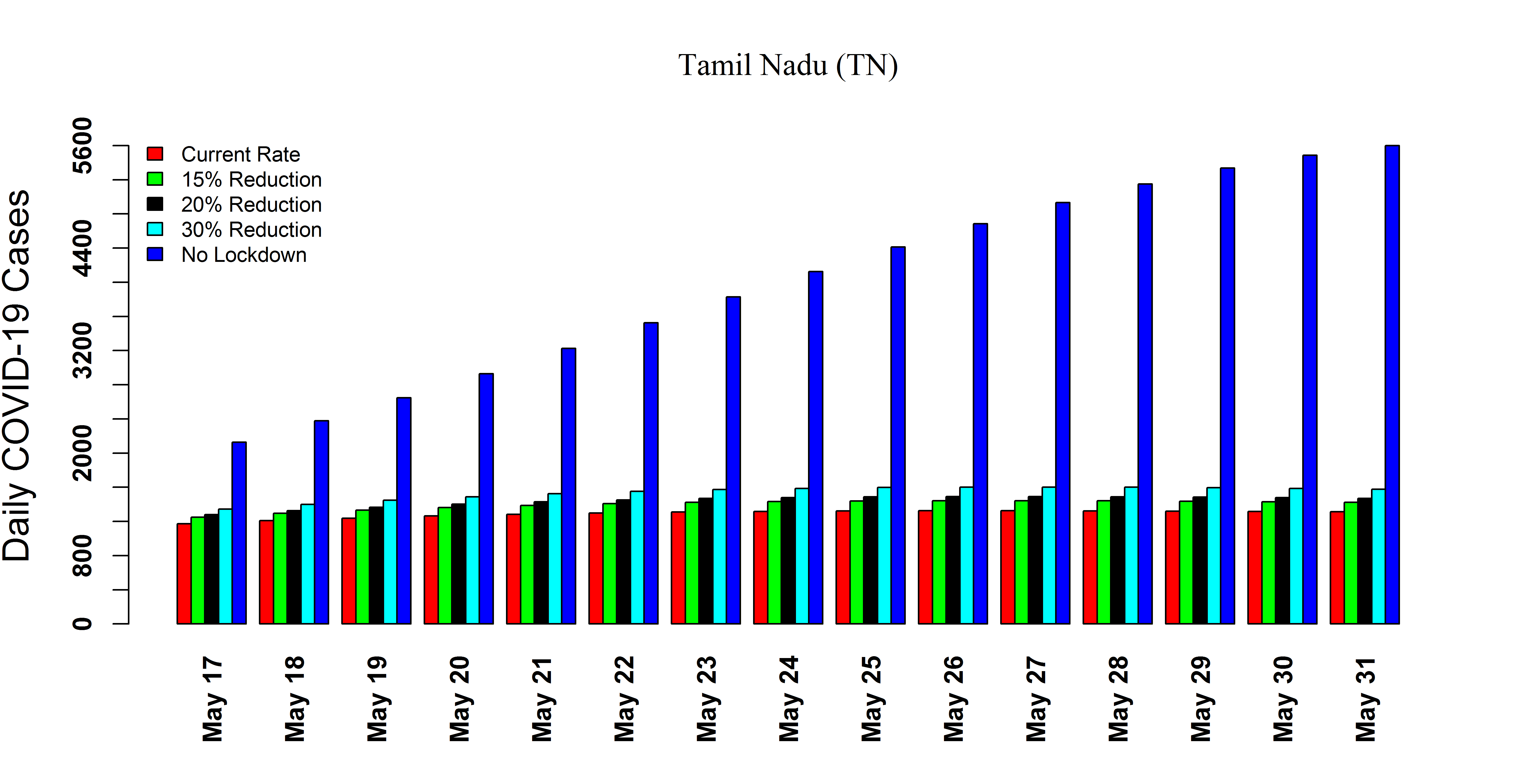}
	\caption{Ensemble model forecast for the daily notified COVID-19 cases in Tamil Nadu during May 17, 2020 till May 31, 2020, under five different social distancing measure. Various legends are same as Fig~\ref{Fig:Prediction-cases-MH}.}
	\label{Fig:Prediction-cases-TN}
\end{figure}

\begin{figure}[ht]
	\captionsetup{width=1.1\textwidth}
	\includegraphics[width=1\textwidth]{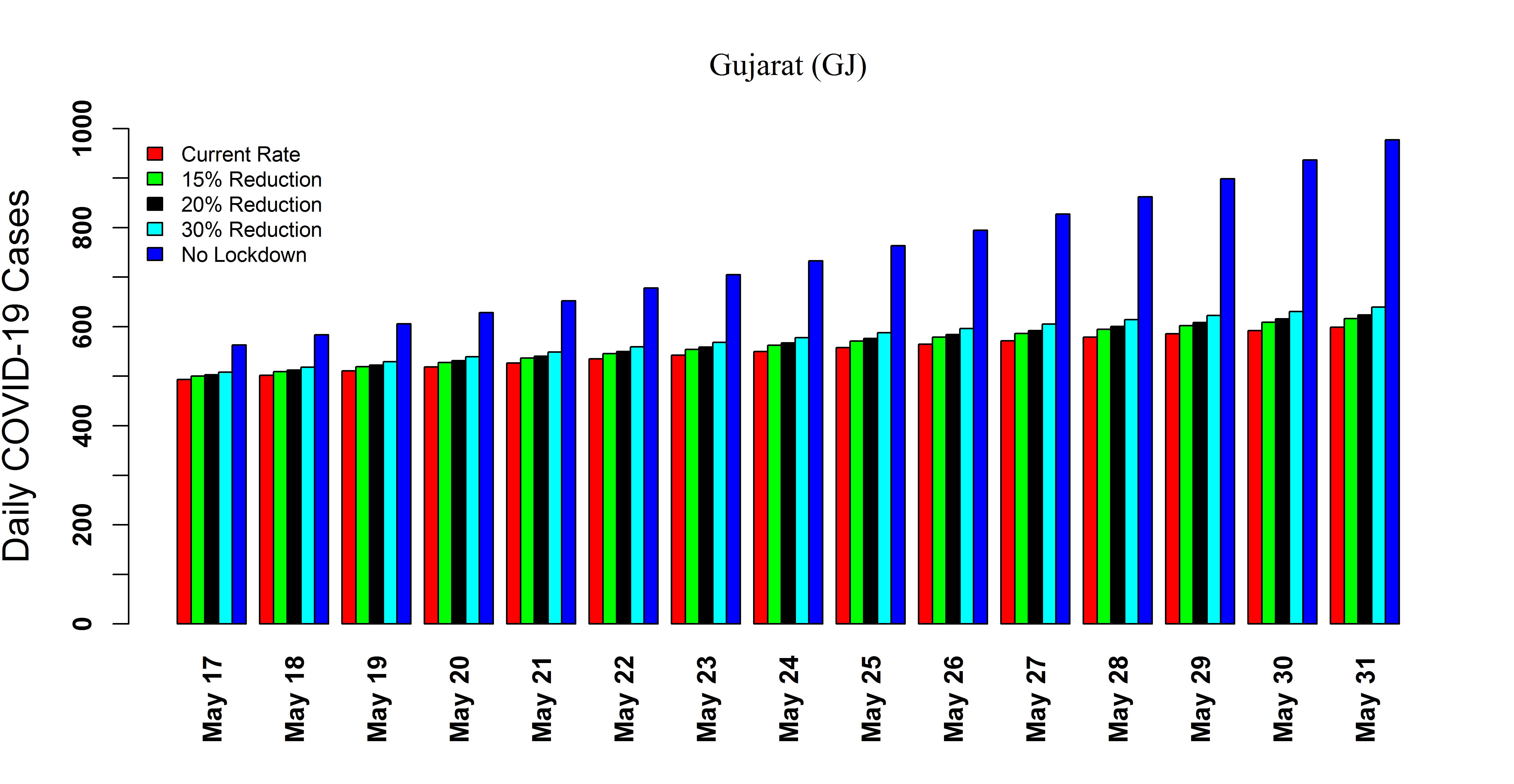}
	\caption{Ensemble model forecast for the daily notified COVID-19 cases in Gujarat during May 17, 2020 till May 31, 2020, under five different social distancing measure. Various legends are same as Fig~\ref{Fig:Prediction-cases-MH}.}
	\label{Fig:Prediction-cases-GJ}
\end{figure}

\begin{figure}[ht]
	\captionsetup{width=1.1\textwidth}
	\includegraphics[width=1\textwidth]{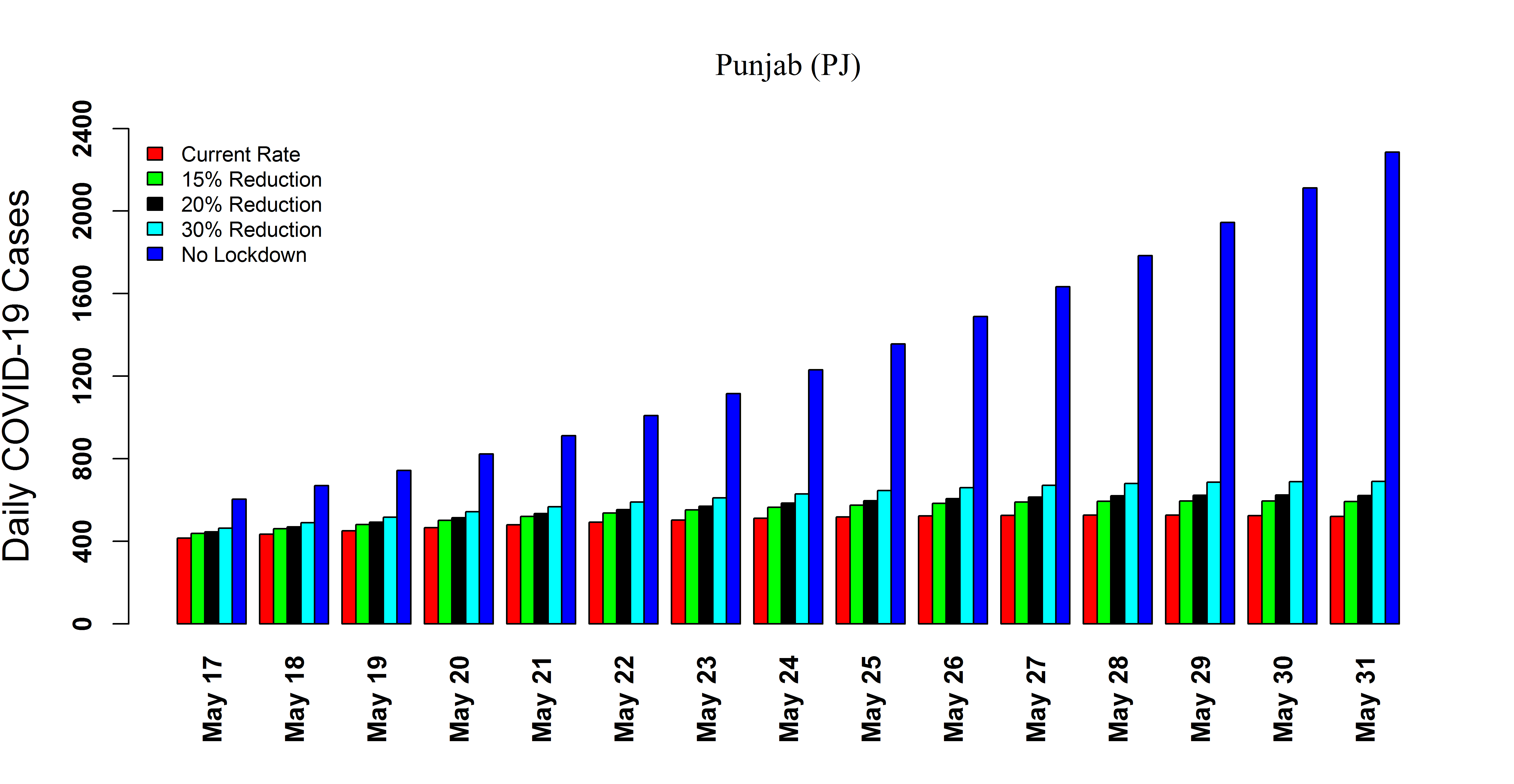}
	\caption{Ensemble model forecast for the daily notified COVID-19 cases in Punjab during May 17, 2020 till May 31, 2020, under five different social distancing measure. Various legends are same as Fig~\ref{Fig:Prediction-cases-MH}.}
	\label{Fig:Prediction-cases-PJ}
\end{figure}

\begin{figure}[ht]
	\captionsetup{width=1.1\textwidth}
	\includegraphics[width=1\textwidth]{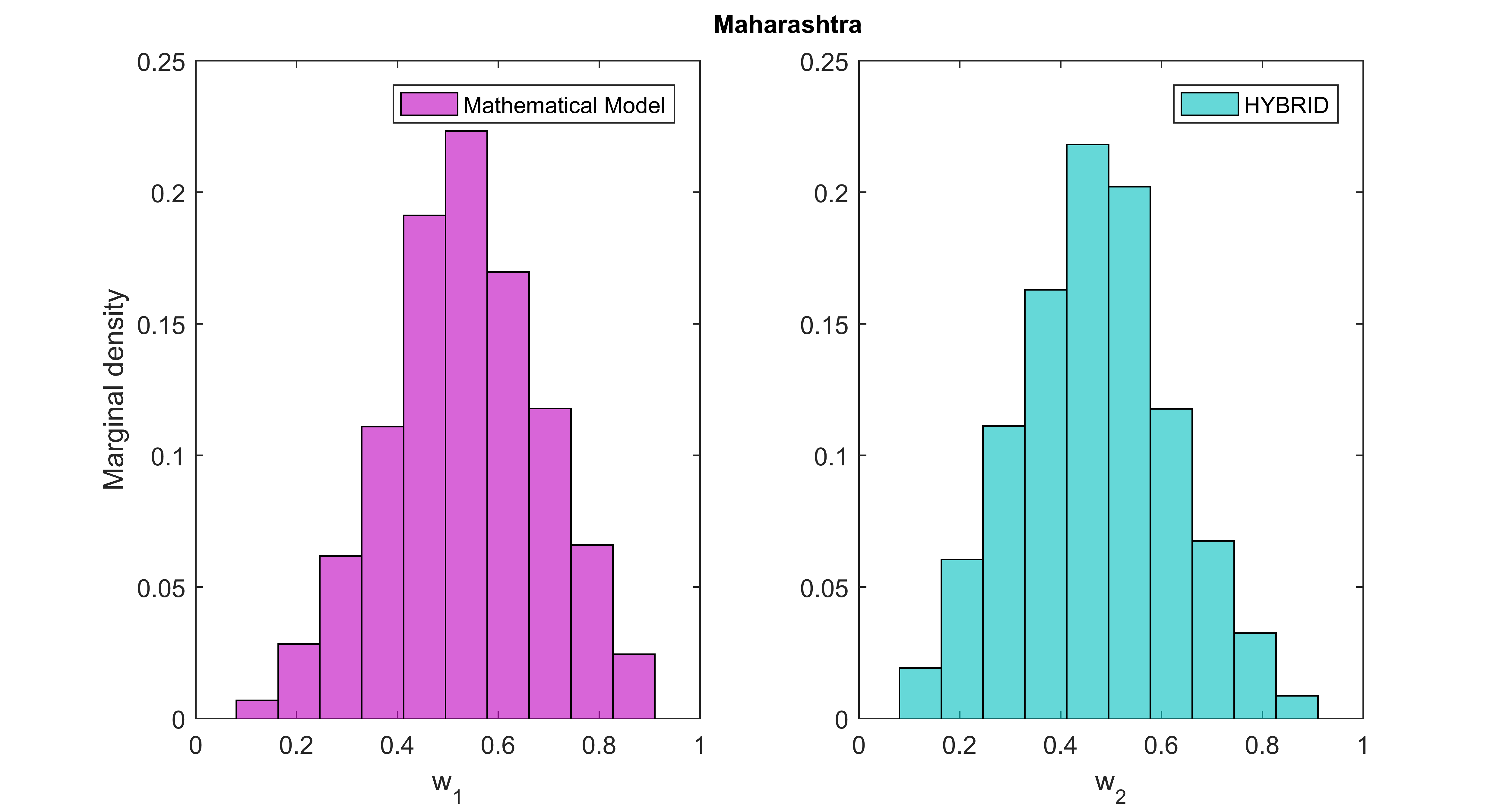}
	\caption{Posterior density of the weights for the mechanistic mathematical model~(\ref{EQ:eqn 2.1} \& \ref{EQ:eqn 3.1}) and the best statistical forecast model (HYBRID), respectively for Maharashtra.}
	\label{Fig:Marginal-distribution-MH}
\end{figure}

\begin{figure}[ht]
	\captionsetup{width=1.1\textwidth}
	\includegraphics[width=1\textwidth]{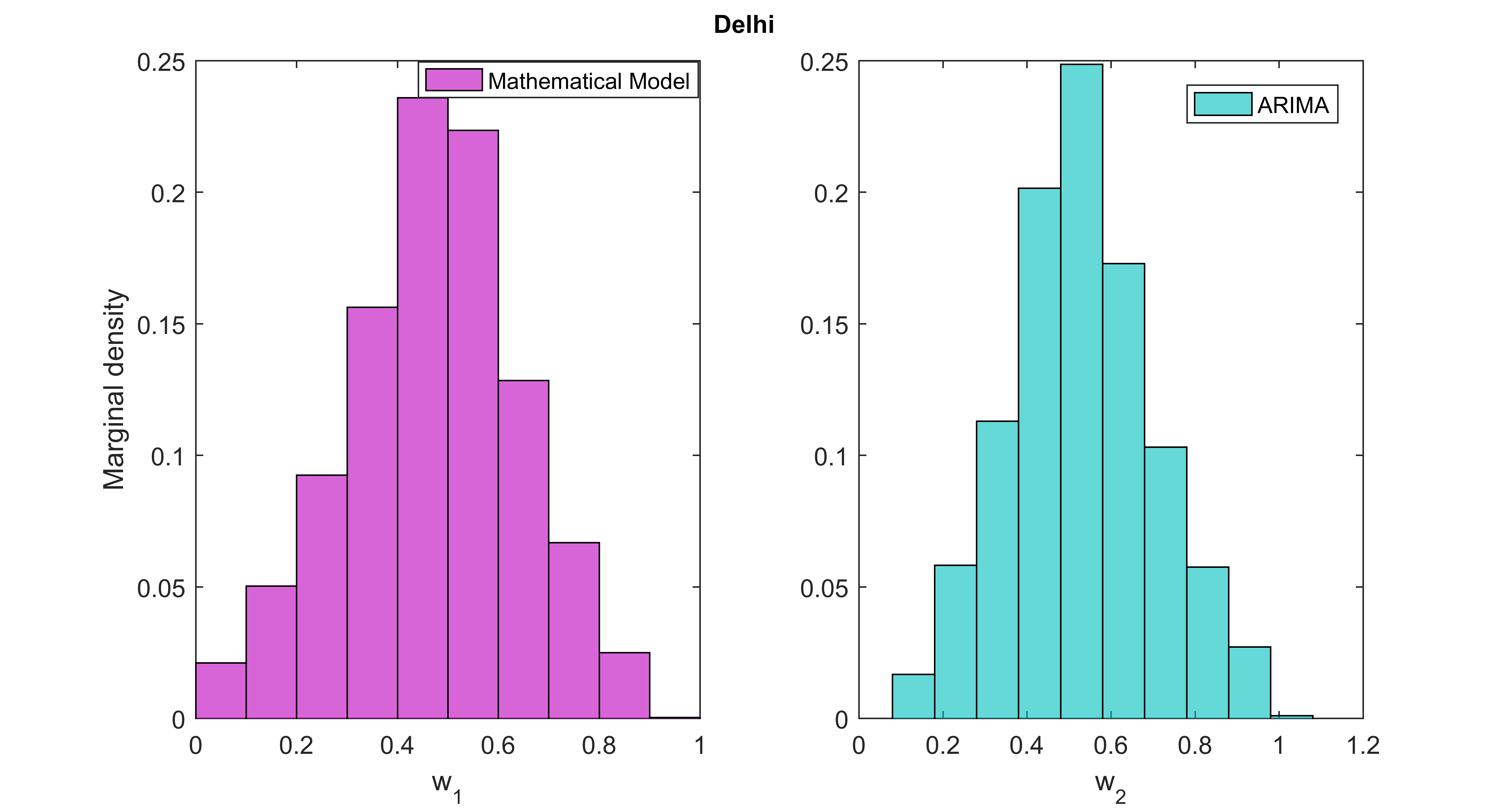}
	\caption{Posterior density of the weights for the mechanistic mathematical model~(\ref{EQ:eqn 2.1} \& \ref{EQ:eqn 3.1}) and the best statistical forecast model (ARIMA), respectively for Delhi}
	\label{Fig:Marginal-distribution-DL}
\end{figure}

\begin{figure}[ht]
	\captionsetup{width=1.1\textwidth}
	\includegraphics[width=1\textwidth]{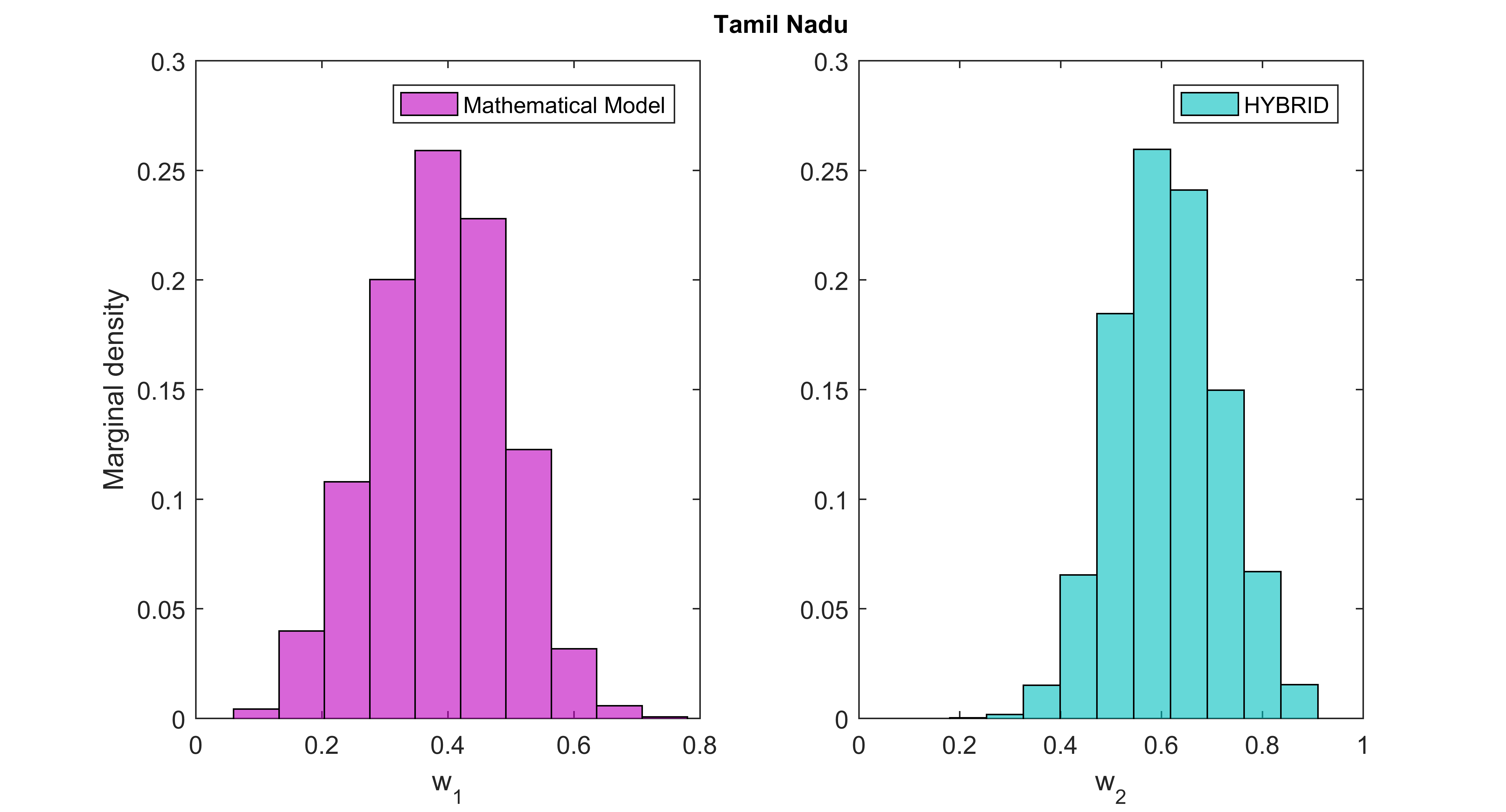}
	\caption{Posterior density of the weights for the mechanistic mathematical model~(\ref{EQ:eqn 2.1} \& \ref{EQ:eqn 3.1}) and the best statistical forecast model (HYBRID), respectively for Tamil Nadu.}
	\label{Fig:Marginal-distribution-TN}
\end{figure}

\begin{figure}[ht]
	\captionsetup{width=1.1\textwidth}
	\includegraphics[width=1\textwidth]{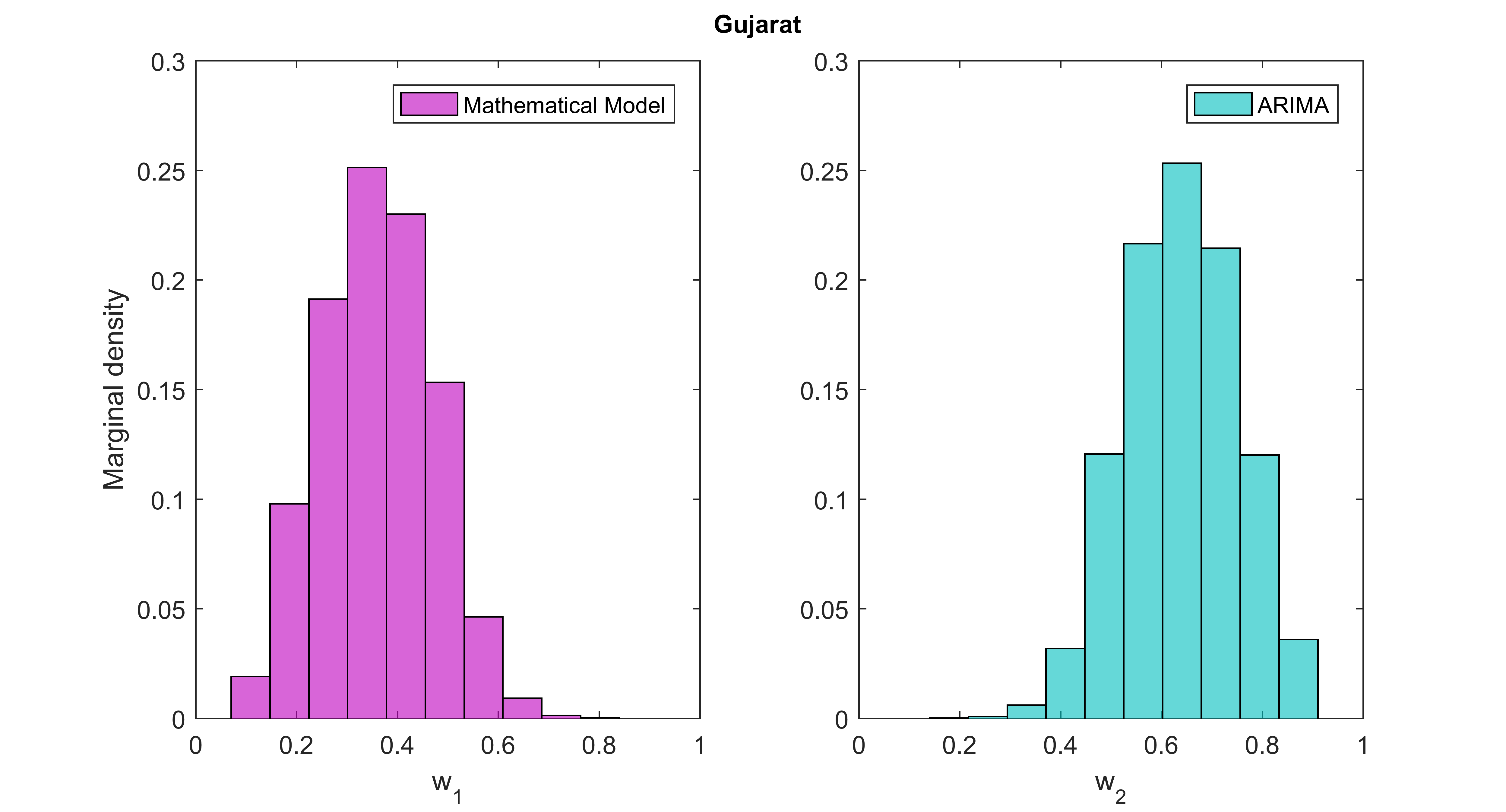}
	\caption{Posterior density of the weights for the mechanistic mathematical model~(\ref{EQ:eqn 2.1} \& \ref{EQ:eqn 3.1}) and the best statistical forecast model (ARIMA), respectively for Gujarat.}
	\label{Fig:Marginal-distribution-GJ}
\end{figure}

\begin{figure}[ht]
	\captionsetup{width=1.1\textwidth}
	\includegraphics[width=1\textwidth]{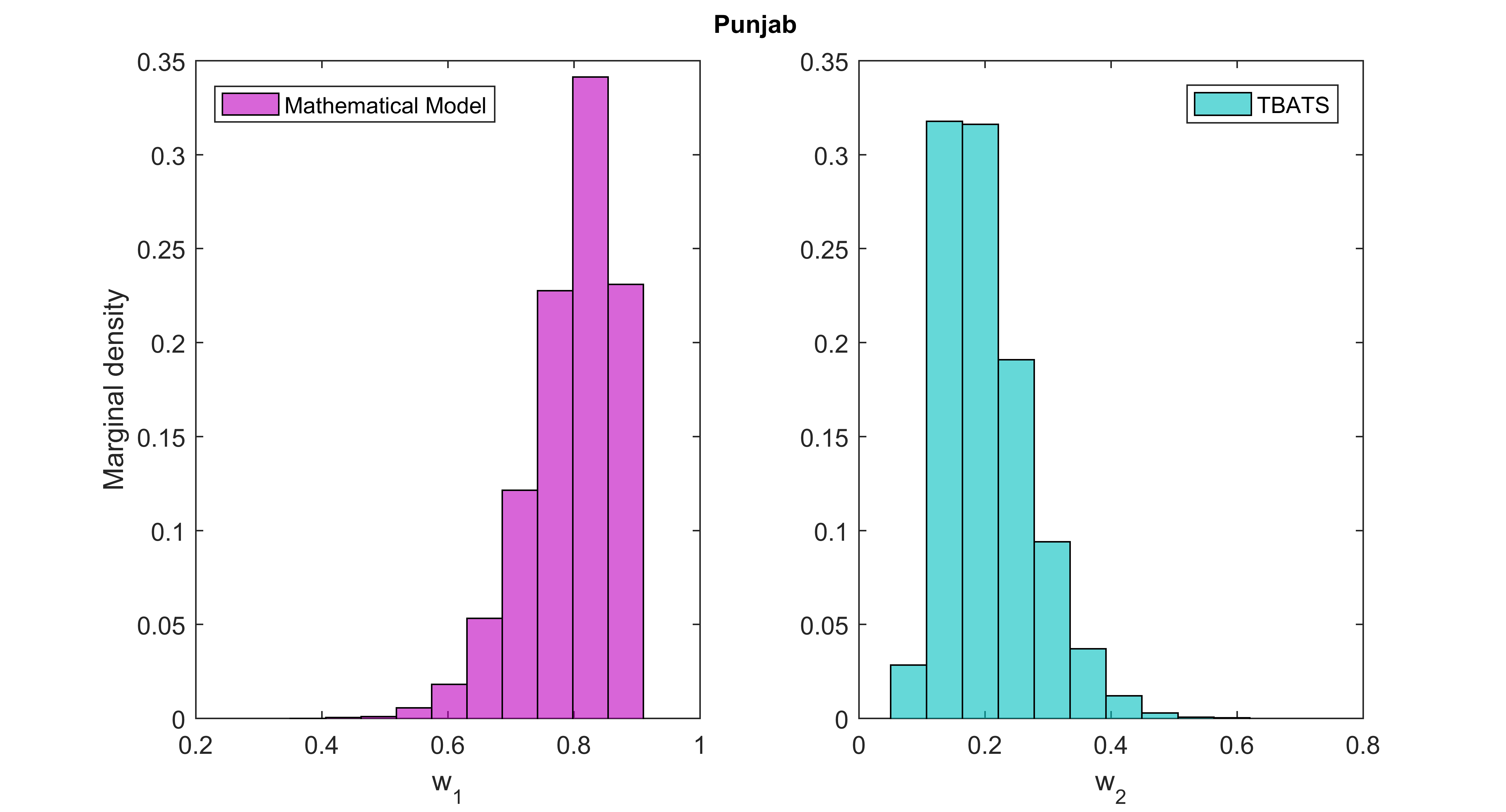}
	\caption{Posterior density of the weights for the mechanistic mathematical model~(\ref{EQ:eqn 2.1} \& \ref{EQ:eqn 3.1}) and the best statistical forecast model (TBATS), respectively for Punjab.}
	\label{Fig:Marginal-distribution-PJ}
\end{figure}

\begin{figure}[ht]
	\captionsetup{width=1.1\textwidth}
	\includegraphics[width=1\textwidth]{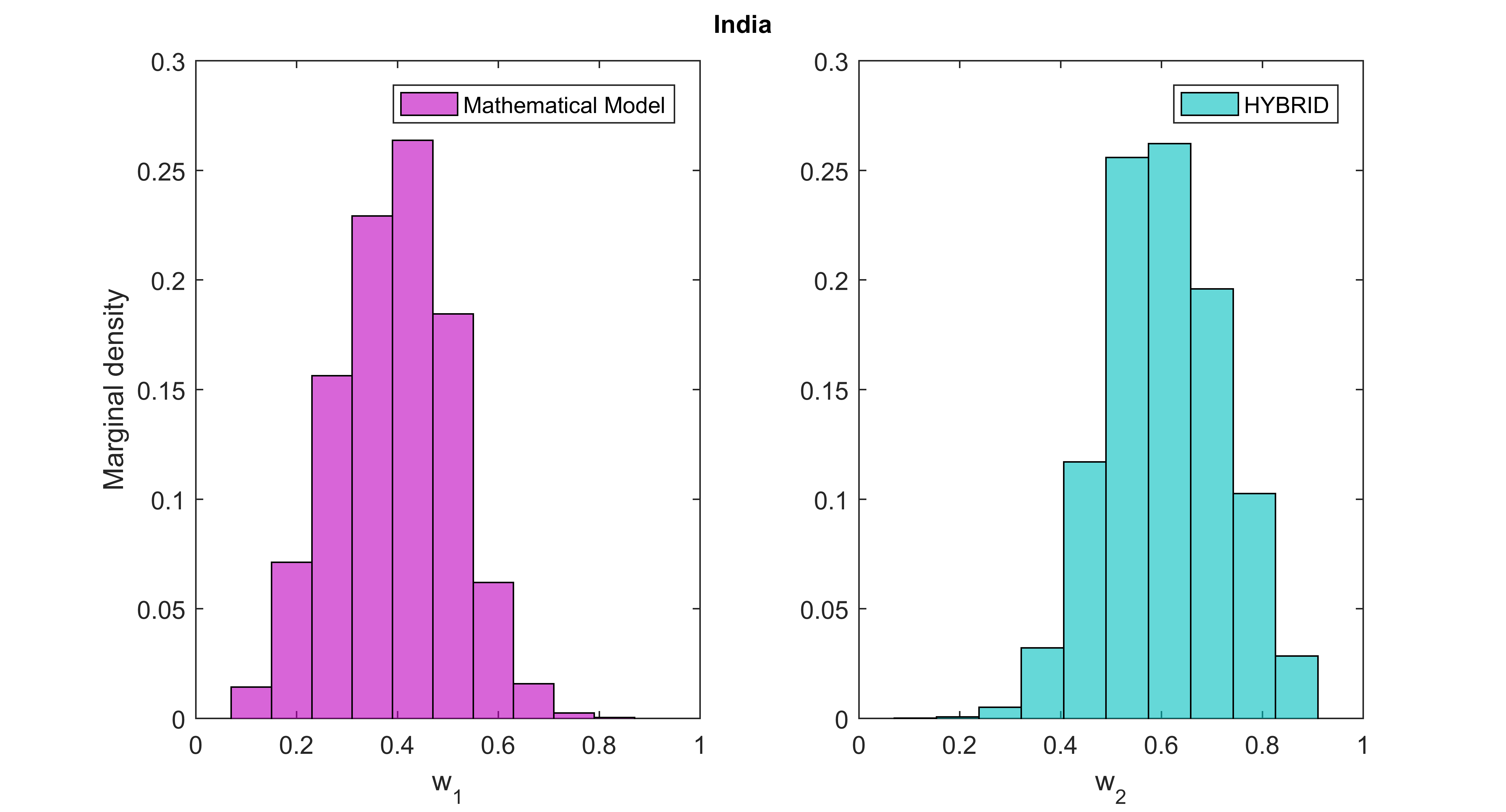}
	\caption{Posterior density of the weights for the mechanistic mathematical model~(\ref{EQ:eqn 2.1} \& \ref{EQ:eqn 3.1}) and the best statistical forecast model (HYBRID), respectively for India.}
	\label{Fig:Marginal-distribution-IND}
\end{figure}

\clearpage
\begin{center}
	\section*{\Large{\underline{Tables}}}
\end{center}

\begin{table}[ht]
	\captionsetup{font=normalsize}
	\captionsetup{width=1.1\textwidth}
	\tabcolsep 16pt		
	\centering
	\caption{Estimated initial state variables of the mathematical model~(\ref{EQ:eqn 2.1}). All data are given in the format \textbf{Estimate (95\% CI)}.}\vspace{0.3cm}
	\begin{tabular}{p{3cm}|ccccccccc} \hline\\
		\textbf{Location} &  $\boldsymbol{S(0)}$ & $\boldsymbol{E(0)}$  & $\boldsymbol{A(0)}$  & $\boldsymbol{I(0)}$ \\ \hline\\
		\textbf{Maharashtra}& $\substack{123243118  \\  \\ (114977907 - 124511349)}$ & $\substack{27.27  \\  \\ (0.51 - 35.06)}$ & $\substack{9531  \\  \\ (9483 - 9989)}$ & $\substack{24.76  \\  \\ (14.17 - 30.26)}$\\\\
		\hline\\
		\textbf{Delhi}& $\substack{16070717  \\  \\ (10346593 - 19737264)}$ & $\substack{14.74  \\  \\ (0.05 - 14.97)}$ & $\substack{437  \\  \\ (92.38 - 4017)}$ & $\substack{282  \\  \\ (7.82 - 325)}$\\\\
		\hline\\
		\textbf{Tamil Nadu}& $\substack{75367306  \\  \\ (70335585 - 79738740)}$ & $\substack{4.80 \\  \\ (0.05 - 4.73)}$ & $\substack{0.40  \\  \\ (0.03 - 2.52)}$ & $\substack{7.61  \\  \\ (1.96 - 10.88)}$\\\\
		\hline\\
		\textbf{Gujarat}& $\substack{66848292  \\  \\ (60218205 - 69744037)}$ & $\substack{7630 \\  \\ (2622 - 9775)}$ & $\substack{26.10  \\  \\ (1.26 - 28.70)}$ & $\substack{14.01  \\  \\ (0.02 - 20.54)}$\\\\
		\hline\\
		\textbf{Punjab}& $\substack{26433551  \\  \\ (26433543 - 26433556)}$ & $\substack{7.83 \\  \\ (0.20- 13.81)}$ & $\substack{19.26  \\  \\ (0.29 - 29.54)}$ & $\substack{8  \\  \\ (1.08 - 15.96)}$\\\\
		\hline\\
		\textbf{India}& $\substack{1226841787  \\  \\ (1219848614 - 1297276832)}$ & $\substack{221104 \\  \\ (21136 - 279521)}$ & $\substack{462997  \\  \\ (27541 - 642484)}$ & $\substack{36043  \\  \\ (3601 - 84059)}$\\\\
		\hline
	\end{tabular}
	\label{Tab:estimated-initial-Table}
\end{table}

\begin{table}[ht]
	\captionsetup{font=normalsize}
	\captionsetup{width=1.1\textwidth}
	\tabcolsep 35pt		
	\centering
	\caption{Goodness of fit (RMSE) for the three statistical forecast model (ARIMA, TBATS and HYBRID), respectively. RMSE for different locations are calculated only for the test period data (May 4, 2020 till May 8, 2020). RMSE values of the best performed statistical forecast model in different locations are shown in red.}\vspace{0.3cm}
	\begin{tabular}{c|ccc} \hline\\
		\textbf{Location} &  \textbf{ARIMA} & \textbf{TBATS}  & \textbf{HYBRID}\\ \hline\\
		\textbf{Maharashtra}& \textbf{179.84368} & \textbf{110.57237} & \textbf{\textcolor{red}{90.04999}}\\\\
		\hline\\
		\textbf{Delhi}& \textbf{\textcolor{red}{92.17782}} & \textbf{96.11452} & \textbf{94.14085}\\\\
		\hline\\
		\textbf{Tamil Nadu}& \textbf{182.5897} & \textbf{244.8254} & \textbf{\textcolor{red}{119.9740}}\\\\
		\hline\\
		\textbf{Gujarat}& \textbf{\textcolor{red}{28.16381}} & \textbf{37.36576} & \textbf{30.92572}\\\\
		\hline\\
		\textbf{Punjab}& \textbf{292.5078} & \textbf{\textcolor{red}{188.9780}} & \textbf{238.9491}\\\\
		\hline\\
		\textbf{India}& \textbf{457.7622} & \textbf{438.1160} & \textbf{\textcolor{red}{368.0725}}\\
		\hline
	\end{tabular}
	\label{Tab:Goodness-of-Fit}
\end{table}

\begin{table}[ht]
	\captionsetup{font=normalsize}
	\captionsetup{width=1.1\textwidth}
	\tabcolsep 10pt		
	\centering
	\caption{Weight estimates for the mechanistic mathematical model~(\ref{EQ:eqn 2.1} \& \ref{EQ:eqn 3.1}) and the best statistical forecast model, respectively. Respective subscript are MH: Maharashtra, DL: Delhi, TN: Tamil Nadu, GJ: Gujarat, PJ: Punjab, and IND: India. $\boldsymbol{w_{1}}$ and $\boldsymbol{w_{2}}$ denote the weights of the COVID-19 mathematical model~(\ref{EQ:eqn 2.1} \& \ref{EQ:eqn 3.1}) and the best statistical forecast model, respectively for a region. All data are provided in the format~\textbf{Estimate (95\% CI)}.}  
	\begin{tabular}{|p{1.8cm}|ccccccccc} \hline\\
		\textbf{Weights} &  \textbf{MH} & \textbf{DL}  & \textbf{TN}  & \textbf{GJ} &  \textbf{PJ} & \textbf{IND} \\ \hline\\
		$\boldsymbol{w_{1}}$& \footnotesize{$\substack{0.48  \\  \\ (0.2243 - 0.8262)}$} & \footnotesize{$\substack{0.5186  \\  \\ (0.1112 - 0.8004)}$} & \footnotesize{$\substack{0.6158  \\  \\ (0.1801 - 0.5848)}$} & \footnotesize{$\substack{0.2009  \\  \\ (0.1554 - 0.5737)}$} & \footnotesize{$\substack{0.8331  \\  \\ (0.6291 - 0.8942)}$} & \footnotesize{$\substack{0.3131  \\  \\ (0.1695 - 0.6149)}$}\\\\
		\hline\\
		$\boldsymbol{w_{2}}$ & \footnotesize{$\substack{0.52  \\  \\ (0.1738 - 0.7757)}$} & \footnotesize{$\substack{0.4814  \\  \\ (0.1996 - 0.8888)}$} & \footnotesize{$\substack{0.3842 \\  \\ (0.4152 - 0.8199)}$} & \footnotesize{$\substack{0.7991  \\  \\ (0.4263 - 0.8446)}$} & \footnotesize{$\substack{0.1669  \\  \\ (0.1058 - 0.3709)}$} & \footnotesize{$\substack{0.6869  \\  \\ (0.3851 - 0.8305)}$}\\
		\hline		
	\end{tabular}
	\label{Tab:estimated-weights}
\end{table}

\end{document}